\documentclass{article}

\usepackage[affil-it]{authblk}
\usepackage[usenames,dvipsnames]{xcolor}
\usepackage{amsfonts}
\usepackage{amsmath,amsthm,amssymb,dsfont}
\usepackage{enumerate}
\usepackage{graphicx}	
\usepackage{subcaption}
\usepackage[margin=3cm]{geometry}
\usepackage{url}
\usepackage{todonotes}
\usepackage{bbm}

\usepackage{tikz}

\usepackage{pifont}
\usepackage{multirow}
\usepackage{makecell}

\usepackage{epsfig}
\usetikzlibrary{shapes.symbols,patterns} % for source symbols
\usetikzlibrary{decorations.pathreplacing,angles,quotes}
\usepackage{pgfplots}
\pgfplotsset{compat=1.10}
\usepgfplotslibrary{fillbetween}

\definecolor{linkblue}{HTML}{001487}
\usepackage{hyperref}
\hypersetup{colorlinks=true,citecolor=linkblue,linkcolor=linkblue,filecolor=linkblue,urlcolor=linkblue,breaklinks=true}

\usepackage{nicefrac}
\usepackage{mathtools}
\usepackage[nameinlink,capitalize,noabbrev]{cleveref}

\usepackage{algorithm}
\usepackage{algorithmic}

\usepackage{mdframed}
\usepackage{aligned-overset}
\theoremstyle{plain}
\newtheorem{theorem}{Theorem}[section]
\newtheorem{theorem-nocolor}{Theorem}[section]
\newtheorem{lemma}[theorem]{Lemma}

\newtheorem{claim}[theorem]{Claim}
\newtheorem{fact}[theorem]{Fact}
\newtheorem{corollary}[theorem]{Corollary}
\newtheorem{proposition}[theorem]{Proposition}

\theoremstyle{definition}
\newtheorem{definition}[theorem]{Definition}
\newtheorem{remark}[theorem]{Remark}
\newtheorem{example}[theorem]{Example}
\newtheorem{question}[theorem]{Question}

\usepackage{tcolorbox}
\tcolorboxenvironment{definition}{}
\tcolorboxenvironment{theorem}{colback=pink!25!white,colframe=pink!100!black}
\tcolorboxenvironment{theorem-nocolor}{colback=white,colframe=lightgray}
\tcolorboxenvironment{lemma}{colback=yellow!5!white,colframe=yellow!75!black}
\tcolorboxenvironment{corollary}{colback=Dandelion!5!white,colframe=Dandelion!75!black}
\tcolorboxenvironment{proposition}{colback=Emerald!5!white,colframe=Emerald!75!black}
\tcolorboxenvironment{claim}{colback=RoyalBlue!5!white,colframe=RoyalBlue!75!black}
\tcolorboxenvironment{conjecture}{colback=red!5!white,colframe=red!75!black}
\tcolorboxenvironment{example}{colback=white,colframe=lightgray}
\tcolorboxenvironment{remark}{colback=white,colframe=lightgray}
\tcolorboxenvironment{question}{colback=SkyBlue!5!white,colframe=lightgray!85!SkyBlue}
\tcolorboxenvironment{fact}{colback=Purple!5!white,colframe=Purple!75!black}
\DeclareRobustCommand{\abbrevcrefs}{%
\Crefname{theorem}{Thm.}{Thms.}%
\Crefname{corollary}{Cor.}{Cors.}%
\Crefname{lemma}{Lem.}{Lems.}%
\Crefname{proposition}{Prop.}{Props.}%
\Crefname{equation}{Eq.}{Eqs.}%
\Crefname{example}{Ex.}{Exs.}%
}

\DeclareRobustCommand{\Cshref}[1]{{\abbrevcrefs\Cref{#1}}}

\newcommand*{\PP}{\mathbb{P}}
\newcommand*{\E}{\mathbb{E}}
\newcommand*{\ee}{\mathrm{e}}

\newcommand*{\cB}{\mathcal{B}}

\newcommand*{\cF}{\mathcal{F}}

\newcommand*{\cH}{\mathcal{H}}

\newcommand*{\cM}{\mathcal{M}}
\newcommand*{\cP}{\mathcal{P}}

\newcommand*{\cS}{\mathcal{S}}
\newcommand*{\cT}{\mathcal{T}}

\newcommand*{\cV}{\mathcal{V}}
\newcommand*{\cW}{\mathcal{W}}
\newcommand*{\cX}{\mathcal{X}}

\newcommand*{\N}{\mathbb{N}}

\newcommand*{\R}{\mathbb{R}}

\newcommand*{\C}{\mathbb{C}}

\newcommand*{\LOCC}{\mathrm{LOCC}}

\newcommand*{\St}{\mathrm{S}}

\newcommand*{\eps}{\varepsilon}

\newcommand*{\id}{\mathrm{id}}
\newcommand*{\poly}{\mathrm{poly}}

\newcommand*{\supp}{\mathrm{supp}}
\newcommand*{\tr}{\mathrm{tr}}
\newcommand*{\ket}[1]{| #1 \rangle}
\newcommand*{\bra}[1]{\langle #1 |}

\newcommand{\proj}[1]{|#1\rangle\!\langle #1|}
\newcommand*{\braket}[1]{\langle #1 \rangle}

\newcommand*{\conv}{\mathrm{conv}}

\newcommand*{\di}{\mathrm{d}} % integration d

\newcommand{\norm}[1]{\left\lVert#1\right\rVert}

\usepackage{stmaryrd}
\newcommand{\contradiction}{\lightning}

\usepackage{float}

 \allowdisplaybreaks

\title{Almost-iid information theory}

 \author{\normalsize Giulia Mazzola$^{1}$, David Sutter$^{2}$, and Renato Renner$^{1}$}
  \affil{\small $^{1}$Institute for Theoretical Physics, ETH Zurich\\
  $^{2}$IBM Research Europe -- Zurich
 }
 \date{}

\begin{document}

\maketitle

\begin{abstract}
Information-theoretic techniques are based on the assumption that resources are well characterized by independent and identically distributed (iid) states. This assumption cannot be justified operationally, since, for example, correlations between subsequent systems emitted by a source cannot be detected by any practical tomographic protocol. Operationally motivated symmetry assumptions still imply, via de Finetti theorems, that the resources are described by almost-iid states. This raises the question: Are almost-iid resources as effective as perfect iid resources for information-processing tasks? Here we address this question and prove that the conditional entropy of almost-iid states asymptotically coincides with that of iid states. As an application, this implies that squashed entanglement is robust for almost-iid states, asymptotically matching its value on iid states.

% Traditional information-theoretic techniques have been developed for independent and identically distributed (iid) states. However, from an operational perspective it is difficult to justify the iid assumption. Recent results such as de Finetti theorems indicate that a weaker almost-iid structure can emerge from well-justified operational symmetry assumptions. Almost-iid states, in which most subsystems are iid, are indistinguishable from iid states by any practical protocol. This raises the question: Are almost-iid resources as effective as perfect iid resources for information-processing tasks? 
% Here, we address this question and prove that the conditional entropy of almost-iid states asymptotically coincides with that of iid states.
% This result further shows that the squashed entanglement is robust for almost-iid states, asymptotically coinciding with its value on iid states.
\end{abstract}

%%%%%%%%%%%%%%%%%%%%%%%%%%%%%%%%%%%%%%%%%%%%%%%%%%%%%%%%%%%%%%%%%%%%%%%%%%%%%%%%%%%%
\section{Introduction}
A common assumption in physics is that the same experiment can be repeated many times independently. More precisely, one often assumes that the outcomes $X_1,\ldots,X_{n+k}$ of running the experiment $n+k$ times are described by independent and identically distributed (iid) random variables. In technical terms, this means that their joint distribution \smash{$P_{X_1^{n+k}}$} is factorized, i.e.,~\smash{$P_{X_1^{n+k}} =  (P_{X})^{n+k}$}.
The Italian mathematician Bruno de Finetti cautioned that the ``most common and misleading error" in probability theory is treating the iid assumption as fundamental~\cite{dF70}. He provided a mathematically convincing argument to operationally justify the iid assumption, relating it to the permutation invariance of the outcomes~\cite{deFinetti37,monari1993introduction}.
For many realistic systems, permutation invariance follows from natural assumptions such as the indistinguishability of subsystems (see~\cite{renner07} for a more detailed discussion).\footnote{Alternatively, permutation invariance can also be enforced by randomly permuting the subsystems.}

Suppose that we have a source that generates random variables $X_1^{n+k}$, where we assume that the underlying joint distribution \smash{$P_{X_1^{n+k}}$} is permutation-invariant and that each random variable takes values in the set $\cX$.
When selecting $n$ of these $n+k$ random variables, i.e.~ignoring $k$ random variables, de Finetti's theorem~\cite{diaconis_freedman80} shows that there exists a probability measure $\mu$ on the set of distributions on $\cX$ such that
\begin{align} \label{eq_deFinetti_classical}
\norm{P_{X_1^n} - \int  (Q_X)^{n}  \mu(\di  Q) }_1 \leq \frac{2 d n}{n+k} \,,
\end{align}
where $\norm{\cdot}_1$ denotes the $\ell^1$-norm and $d = |\cX|$.
This result justifies that when ignoring $k$ random variables the remaining ones are approximately a convex combination of iid random variables. 

De Finetti theorems have been generalized to the quantum case. In~\cite{dFFuchs02} a quantum de Finetti theorem has been obtained as an implication of the classical result in the case where $n = \infty$.
This yields a result that is applicable under the assumption of having an infinite number of samples. Later in~\cite{deFinetti_KR_05,deFinetti_CKMR_07} quantum de Finetti theorems have been derived that work for a finite number of samples --- analogous to~\cref{eq_deFinetti_classical}. For any state $\ket{\Phi^{(n+k)}}$ on $\cH^{\otimes n+k}$ that is symmetric (i.e.~invariant under permutations of the subsystems), there exists a probability measure $\mu$ on the unit sphere $\cB(\cH):=\{\ket{\theta} \in \cH : \norm{\theta} =1\}$ such that
\begin{align} \label{eq_deFinetti_quantum}
\norm{ \tr_{k}[\proj{\Phi^{(n+k)}}] - \int  \proj{\theta}^{\otimes n}   \mu(\di \theta) }_1 \leq \frac{2 d^2 n}{n+k} =: \eps_d(n,k)\,,
\end{align}
where $\norm{\cdot}_1$ is the trace norm, $\tr_k[\cdot]$ denotes the partial trace of $k$ subsystems, and $d = \dim(\cH)$.
It is worth mentioning that the results stated in~\cref{eq_deFinetti_classical,eq_deFinetti_quantum} are essentially tight~\cite{diaconis_freedman80, deFinetti_CKMR_07}.

The main obstacle with these de Finetti results (classical and quantum) is that they require $k$ to be large. Specifically, if we want the error term $\eps_d(n,k)$ to vanish in the limit $n \to \infty$, we need $k$ to be superlinear in $n$, i.e.~$k=\omega(n)$.\footnote{Note that, by definition, $f(n) = \omega(g(n)) \iff g(n) = o(f(n))$. Hence, $k=\omega(n) \iff n=o(k)$, or in other words, $\lim_{n \to \infty} \frac{n}{k} = 0$.}
Therefore, we need to ignore a large fraction of our data. This may be justified in certain scenarios where the number of data is by default huge\footnote{For example, random samples of a coin toss, which in principle can be repeated arbitrarily many times.}, however, it can be prohibitive in other settings.\footnote{For example, in the setting of quantum key distribution.}    
A second, but often less drastic, disadvantage is that, by choosing $k=\poly(n)$, the error term $\eps_d(n,k)$ is vanishing at a slow rate of order $1/\poly(n)$.

The insight of~\cite{renner_phd,renner07} was that we can overcome these limitations when relaxing the iid assumption.
More precisely, the \emph{exponential de Finetti theorem}~\cite{renner_phd,renner07} shows that there exist a probability measure $\nu$ on the unit sphere $\cB(\cH)$ and a family \smash{$\{\ket{\Psi_{r,\theta}^{(n)}}\}_{\theta}$} of almost-iid states in $\ket{\theta}$ with a defect of size $r$ such that
\begin{align} \label{eq_general_dF}
\norm{\tr_{k}[\proj{\Phi^{(n+k)}}] - \int \proj{\Psi_{r,\theta}^{(n)}} \nu(\di \theta)}_1 \leq 3 k^d \exp\Big(-\frac{k(r+1)}{n+k}\Big)=:\eps'_d(n,k) \, .
\end{align}
The precise statement is given in~\cref{thm_exp_deFinetti_quantum}.
Almost-iid states in $\ket{\theta}$ are superpositions of states that are equal to $\ket{\theta}^{\otimes n-r}$ on $n-r$ subsystems and arbitrary on the remaining $r$ subsystems. Note that in each element of the superposition the positions of the defects may be different. Usually, the number of defects is sublinear in $n$, i.e.~$r=o(n)$. The precise definition of almost-iid states is given in~\cref{sec_almost_iid_states}.
\Cref{eq_general_dF} has a drastically better parameter scaling compared to~\cref{eq_deFinetti_quantum}. It allows us to choose $k$ sublinear in $n$, i.e.~$k=o(n)$, and still have an error term that vanishes in the limit $n \to \infty$ --- even at an exponential rate.
To see this, note that, for example, by choosing \smash{$k=n^{\frac{3}{4}}$} and $r=\sqrt{n}$, we obtain \smash{$\eps'_d(n,k)=3\, n^{\frac{3d}{4}} \exp(-\frac{n^{5/4}+n^{3/4}}{n+n^{3/4}})=O(\exp(-n^{\frac{1}{4}}))$}. 

While the exponential de Finetti theorem justifies the importance of almost-iid states, it is natural to ask if these states are also relevant in a purely classical scenario. One may define almost-iid distributions as a convex combination of an iid distribution on $n-r$ subsystems and an arbitrary distribution on the remaining $r$ subsystems, where the position of the defects can be different in each element of the convex combination. For such distributions, we show that no classical equivalent to~\cref{eq_general_dF} exists. In other words, it is not possible to approximate a classical permutation-invariant distribution by a convex combination of almost-iid distributions for $k$ sublinear in $n$. This is made precise in~\cref{sec_exp_deFinetti}. 

The lack of a classical exponential de Finetti theorem shows that almost-iid states are considerably more powerful than classical almost-iid distributions. This is due to the allowed superpositions which can contain long-range entanglement. For that reason, almost-iid states are more complicated to analyze. In particular, it is unclear if certain (robust) applications could behave differently for almost-iid and iid states. 
However, it is expected that they should not behave differently for practical purposes.
Quantum tomography, which is used to infer the characteristics of a system, cannot distinguish between almost-iid and iid states for practically feasible scenarios. We refer to~\cref{prop_statistics} for a mathematically rigorous statement.

Given the operational relevance of almost-iid states (ensured by the exponential de Finetti theorem), it is crucial to understand which applications behave asymptotically equally under almost-iid and iid states. In this work, we introduce the definition of mixed almost-iid states and prove that for conditional entropies, there is no asymptotic difference between almost-iid and iid states. 
For $n\in \N$ and $\rho_{A_1^n B_1^n}$ a mixed almost-iid state in $\sigma_{AB}$ with a defect of size $r=o(n)$, we show that
\begin{align} \label{eq_result}
 \frac{1}{n} H(A_1^n|B_1^n)_{\rho} = H(A|B)_{\sigma} + \frac{o(n)}{n} \, ,
\end{align}
where $H(A|B)_{\sigma}:=H(AB)_{\sigma}-H(B)_{\sigma}$ is the conditional entropy and $H(B)_{\sigma} := -\tr[\sigma_B \log \sigma_B]$ for $\sigma_B=\tr_{A}[\sigma_{AB}]$ is the von Neumann entropy.
The exact result is given in~\cref{thm_entropy_of_almost_iid}. 
%Furthermore, we show in~\cref{prop_no_classical_exp_dF} that no classical exponential de Finetti theorem can exist in the sense explained above.
%In~\cref{sec_entag_meas}, we define the entanglement cost and distillation for almost-iid states and prove various properties of these entanglement measures.
To prove our results, we use information-theoretic tools that have been developed to go beyond the traditionally assumed iid structure~\cite{renner_phd,dr09,marco_book}.  

We conclude in~\cref{sec_entag_meas} with a discussion of whether popular entanglement measures asymptotically coincide for almost-iid and iid states. We prove that this holds for the squashed entanglement (see~\cref{cor_Squashed_is_robust}); however, it remains an open question for other entanglement measures such as the distillable entanglement, the entanglement cost, and the relative entropy of entanglement.  

%%%%%%%%%%%%%%%%%%%%%%%%%%%%%%%%%%%%%%%%%%%%%%%%%%%%%%%%%%%%%%%%%%%%
%%%%%%%%%%%%%%%%%%%%%%%%%%%%%%%%%%%%%%%%%%%%%%%%%%%%%%%%%%%%%%%%%%%%
%%%%%%%%%%%%%%%%%%%%%%%%%%%%%%%%%%%%%%%%%%%%%%%%%%%%%%%%%%%%%%%%%%%%
\section{Almost-iid states} \label{sec_almost_iid_states}
In this section, we formally define almost-iid states and discuss some of their properties.
Historically~\cite{renner_phd}, almost-iid states were introduced as pure and symmetric states with additional structure, as explained in the following. This family of states appears, for example, in exponential de Finetti theorems such as those presented in~\cref{eq_general_dF}. However, we will relax the definition from~\cite{renner_phd} to capture more general families of states, including mixed ones, that have a similar almost-iid structure.

Let $\cS_n$ be the set of permutations on $\{ 1,... ,n\}$ and let $\St(\cH)$ denote the set of density matrices on a Hilbert space $\cH$ whose dimension is denoted by $d$.
The symmetric subspace is given by $\mathrm{Sym}^n(\cH):= \mathrm{span} \{ \ket{\phi}^{\otimes n}: \ket{\phi} \in \cH \}$. For a fixed $\ket{\theta}$, let $\cV(\cH^{\otimes n}, \ket{\theta}^{\otimes m}):=\{\pi(\ket{\theta}^{\otimes m} \otimes \ket{\Omega^{(n -m)}}): \pi \in \cS_n, \ket{\Omega^{(n -m)}} \in \cH^{\otimes n -m} \}$ and 
\begin{align}
\mathrm{Sym}^n(\cH, \ket{\theta}^{\otimes m}):= \mathrm{Sym}^n(\cH) \cap \mathrm{span} \, \cV(\cH^{\otimes n}, \ket{\theta}^{\otimes m}) \, .
\end{align}
Let $n,r \in \N$ such that $r \leq n$ and $\ket{\theta} \in \cH$. In~\cite{renner_phd}, an $\binom{n}{r}$-\emph{almost-iid state in} $\ket{\theta}$ was defined as a pure state $\ket{\Psi^{(n)}} \in \mathrm{Sym}^n(\cH, \ket{\theta}^{\otimes n-r})$.

However, for pure states outside the symmetric subspace or mixed states, this definition needs to be relaxed in order to capture states that intuitively have an almost-iid structure such as those mentioned in~\cref{ex_almost_iid} below. We therefore next introduce a novel definition for mixed almost-iid states.
\begin{definition}[Almost-iid states] \label{def_almost_product_state_mixed}
Let $\cH_A$ be a Hilbert space, $\sigma_A \in \St(\cH_A)$, and $n,r \in \N$ such that $r \leq n$. Then, $\rho_{A_1^n} \in \St(\cH_A^{\otimes n})$ is called a $\binom{n}{r}$-\emph{almost-iid state in} $\sigma_A $ if there exists a purification $\ket{\theta}_{AE}$ of $\sigma_A$ and an extension $\rho_{A_1^n E_1^n}$ of $\rho_{A_1^n}$ such that
\begin{enumerate}[(i)]
\item $\rho_{A_1^n E_1^n}$ is permutation-invariant with respect to $(A_i,E_i) \leftrightarrow (A_j,E_j)$;\label{it_first_def}
\item $\supp(\rho_{A_1^n E_1^n}) \subseteq  \mathrm{span}\,\cV(\cH_{AE}^{\otimes n},\ket{\theta}_{AE}^{\otimes n-r})$. \label{it_second_def}
\end{enumerate}
We then write  $\rho_{A_1^n} \in \St^n(\cH_A, \sigma_A^{\otimes n-r})$.
\end{definition}
\noindent
Here, $\supp(X)$ denotes the support of a linear operator $X$ and is defined as $\supp(X)= \ker(X)^\perp$.
We next discuss a few pedagogical examples of mixed almost-iid states.
\begin{example} \label{ex_almost_iid}
A few simple instances of almost-iid states include:
\begin{enumerate}[(a)]
\item Tensor power states $\rho_{A_1^n}=\sigma_A^{\otimes n}$. These are almost-iid states with defect size $r=0$, i.e.~$\rho_{A_1^n} \in \St^n(\cH_A, \sigma_A^{\otimes n})$.
\item Convex combinations of tensor power states with a small number of defects, i.e.,~states of the form \smash{$\rho_{A_1^n}=\frac{1}{n!} \sum_{\pi \in \cS_n} \pi (\sigma_A^{\otimes n-r} \otimes \omega_{A_1^r}) \pi^\dagger$}, where $\omega_{A_1^r}$ denotes an arbitrary density operator of size $r$ representing the defects. We have $\rho_{A_1^n} \in \St^n(\cH_A, \sigma_A^{\otimes n-r})$. To see this, let $\ket{\theta}_{AE}$ and \smash{$\ket{\omega}_{A_1^r E_1^r}$} be purifications of $\theta_A$ and \smash{$\omega_{A_1^r}$}, respectively. Then the extension
\begin{align}
\rho_{A_1^n E_1^n} = \frac{1}{n!} \sum_{\pi \in \cS_n} \pi \big( \proj{\theta}^{\otimes n -r}_{AE} \otimes \proj{\omega}_{A_1^r E_1^r} \big) \pi^\dagger 
\end{align}
of \smash{$\rho_{A_1^n}$} clearly satisfies property~\eqref{it_first_def} since it is permutation-invariant. It also satisfies the property~\eqref{it_second_def} as it can be written as \smash{$\rho_{A_1^n E_1^n} = \sum_{i,j \in \cT} \beta_{i,j} \ket{\Psi_i} \bra{\Psi_j}$} for $\beta_{i,j} = \frac{1}{n!} \delta_{i,j}$ and  $\ket{\Psi_i} \in \cV(\cH_{AE}^{\otimes n},\ket{\theta}_{AE}^{\otimes n-r})$ for all $i \in \cT$.
\item The antisymmetric Bell state $\rho_{A_1^2} = \proj{\Psi^-}_{A_1^2}$ with $\ket{\Psi^-}_{A_1^2} := \frac{1}{\sqrt{2}}(\ket{0}\ket{1}_{A_1^2} - \ket{1}\ket{0}_{A_1^2})$ is a $\binom{2}{1}$-almost-iid state in $\ket{0}_A$. \label{it_counter_ex_fernando}
\end{enumerate}
\end{example}
We want to emphasize that the class of almost-iid states contains many more families than the ones presented in~\cref{ex_almost_iid}. In particular, it allows for superpositions of product states with a  small number of defects. These superpositions are crucial for an exponential de Finetti theorem as in~\cref{eq_general_dF} to hold. On the other hand, classical almost-iid distributions have a simpler structure as they do not have superpositions and hence are just convex combinations of almost-product distributions. This is discussed in~\cref{sec_exp_deFinetti}.

\begin{remark} \label{rmk_fernando_def}
A more restrictive definition of mixed almost-iid states was introduced in~\cite{brandao_Stein_10} by calling $\rho_{A_1^n} \in \St(\cH_A^{\otimes n})$ an $\binom{n}{r}$-\emph{almost-iid state in} $\sigma_A $ if there exists a purification $\ket{\Psi^{(n)}}_{A_1^n E_1^n}$ such that $\ket{\Psi^{(n)}} \in \mathrm{Sym}^n(\cH_A \otimes \cH_E, \ket{\theta}_{AE}^{\otimes n-r})$ for some purification $\ket{\theta}_{AE}$ of $\sigma_A$. Such states clearly satisfy the assumptions of~\cref{def_almost_product_state_mixed}. However,~\cref{def_almost_product_state_mixed} is strictly broader as it includes states that do not meet the definition given in~\cite{brandao_Stein_10}. One such example is given by~\cref{ex_almost_iid}~\eqref{it_counter_ex_fernando} as explained in~\cref{app_justification_no_fernando}.
\end{remark}

We call $\rho_{A_1^n}$ a $\binom{n}{r}$-\emph{generalized-almost-iid-state in} $\sigma_A$ if it satisfies~\eqref{it_second_def} but not necessarily~\eqref{it_first_def}. We then write  $\rho_{A_1^n} \in \bar \St^n(\cH_A, \sigma_A^{\otimes n-r})$. Trivially, we have $\St^n(\cH_A, \sigma_A^{\otimes n-r}) \subseteq \bar \St^n(\cH_A, \sigma_A^{\otimes n-r})$.
With the following trace-preserving completely positive map
\begin{align}
\mathrm{PERM}: X \mapsto \frac{1}{n!}\sum_{\pi \in \cS_n} \pi X \pi^\dagger \, ,
\end{align}
which symmetrizes the input, it is possible to convert a generalized-almost-iid-state into an almost-iid state in the sense that
\begin{align}
\rho_{A_1^n} \in \bar \St^n(\cH_A, \sigma_A^{\otimes n-r}) \implies \mathrm{PERM}(\rho_{A_1^n}) \in \St^n(\cH_A, \sigma_A^{\otimes n-r}) \, .
\end{align}

\begin{remark}[Properties of almost-iid states]\ \label{rmk_properties_almost_iid_states}
Let $\rho_{A_1^n} \in \St^n(\cH_A,\sigma_A^{\otimes n-r})$ be a mixed almost-iid state according to~\cref{def_almost_product_state_mixed}. Then it satisfies the following properties:
\begin{enumerate}[(a)]
\item For $\emph{any}$ purification $\ket{\theta}_{AE}$ of $\sigma_A$ there exists an extension $\rho_{A_1^n E_1^n}$ of $\rho_{A_1^n}$ that satisfies~\eqref{it_first_def} and~\eqref{it_second_def}. See~\cref{eq_def_almost_iid_mixed_states}. 
%\item In case that $\sigma$ and $\rho$ are pure,~\cref{def_almost_product_state_mixed} reduces to~\cref{def_almost_product_state} since the extension can be chosen to be trivial.
\item \smash{$\rho_{A_1^n}$} is permutation-invariant. 
\item There exists an orthonormal basis $\{\ket{\Psi_t}\}_{t \in \cT}$ of $\mathrm{span}\,\cV(\cH_{AE}^{\otimes n},\ket{\theta}^{\otimes n-r})$ with vectors $\ket{\Psi_t} \in \cV(\cH_{AE}^{\otimes n},\ket{\theta}^{\otimes n-r})$ and with
\begin{align} \label{eq_sizeT}
    |\cT| \leq \binom{n}{r}\,d_{AE}^{\,r} \leq 2^{n h(r/n)}\, d_{AE}^{\,r},
\end{align} 
where $h(x):=-x\log x - (1-x)\log(1-x)$ for $x \in [0,1]$ is the binary entropy function. With respect to this basis, condition~\eqref{it_second_def} can be rewritten  as
\begin{align} \label{eq_almost-iid-form}
\rho_{A_1^n E_1^n} = \sum_{i,j \in \cT} \beta_{i,j} \ket{\Psi_i} \bra{\Psi_j} \, ,
\end{align}
for coefficients $\beta_{i,j} \in \C$ with $\beta_{i,i} \in [0,1]$ and $\sum_{i \in \cT} \beta_{i,i}=1$. See~\cref{eq_ONB_almost_iid_mixed_states}.
\item  $\rho_{A_1^{n+m}} \in \St^{n+m}(\cH_A, \sigma_A^{\otimes n+m-r})$  implies $\rho_{A_1^n} \in \St^n(\cH_A, \sigma_A^{\otimes n-r})$ for any $m,n\in \N$. See~\cref{fact_Giulia_proof}.

%    \item $\ket{\Psi}_{A_1^n E_1^n} \in \mathrm{Sym}(\cH_{AE}^{\otimes n}, \ket{\theta}_{AE}^{\otimes n-r})$ implies $\rho_{A_1^n} \in \St(\cH_A^{\otimes n},\sigma_A^{\otimes n-r})$ for $\sigma_A = \tr_E[\proj{\theta}_{AE}]$ and $\rho_{A_1^n} = \tr_{E_1^n}[\proj{\Psi}_{A_1^n E_1^n}]$.
\item  $\rho_{A_1^n B_1^n} \in \St^n(\cH_{AB},\sigma_{AB}^{\otimes n-r})$ implies $\rho_{A_1^n} \in \St^n(\cH_{A},\sigma_{A}^{\otimes n-r})$. This follows directly from the definition of mixed almost-iid states.
\item For any $\rho_{A_1^n} \in \bar \St^n(\cH_A, \sigma_A^{\otimes n-r})$ we have for any $s \in \N$ that $(\rho_{A_1^n})^{\otimes s} \in \bar \St^{ns}(\cH_A, \sigma_A^{\otimes ns-rs})$. See~\cref{lem_PERM_almost_iid}.
\item  The set of almost-iid states is convex. This follows directly from the definition of mixed almost-iid states.
%\item For $\rho_{A_1^n} \in \St(\cH_A^{\otimes n}, \sigma_A^{\otimes n-r})$ and $\bar \rho_{B_1^n} \in \St(\cH_B^{\otimes n}, {\bar \sigma_B}^{\otimes n- \bar r})$ we have $\rho_{A_1^n} \otimes \bar \rho_{B_1^n} \in \St(\cH_{AB}^{\otimes n}, (\sigma_A \otimes \bar \sigma_B)^{\otimes n- r- \bar r})$.\footnote{A proof is given in~\cref{fact_tensor_product}.}
\end{enumerate}
\end{remark}

% For any $\ket{\Psi^{(n)}} \in \mathrm{Sym}(\cH^{\otimes n}, \ket{\theta}^{\otimes n-r})$ there exists a family $\{\ket{\Psi_t} \}_{t \in \cT}$ of orthonormal vectors from $\cV(\cH^{\otimes n},\ket{\theta}^{\otimes n-r})$ such that
% \begin{align} \label{eq_ONB_dec}
% \ket{\Psi^{(n)}} = \sum_{t \in \cT} \gamma_t \ket{ \Psi_t} \, ,
% \end{align}
% with $\sum_{t \in \cT } |\gamma_t|^2 =1$. By performing a measurement in the basis $\{\ket{\Psi_t} \}_{t \in \cT}$, which is sometimes also called performing a pinching operation~\cite{Sutter_book} $\cP: X \mapsto \sum_{t \in \cT} \proj{\Psi_t} X \proj{\Psi_t}$, we obtain
% \begin{align} \label{eq_classical_almost_iid_state}
% \cP(\proj{\Psi^{(n)}}) = \sum_{t \in \cT} |\gamma_t|^2 \proj{\Psi_t} \, ,
% \end{align}
% which is a convex mixture of almost-product states.
% \DS{Is it clear that any classical almost-iid distribution has this form? Or why can we restrict ourselves to only the measurements as described above? Is it because of the de Finetti?}

In the analysis of almost-iid states, the representation in~\cref{eq_almost-iid-form} is particularly useful because it exploits the product structure of the space $\cV(\cH_{CE}^{\otimes n},\ket{\theta}_{CE}^{\otimes n-r})$. However, it is sometimes preferable to go to a purified description, i.e., to work with a purification of~\cref{eq_almost-iid-form}. The following lemma establishes that the underlying structure of $\cV(\cH_{CE}^{\otimes n},\ket{\theta}_{CE}^{\otimes n-r})$ persists even in this purified formulation.
\begin{lemma}[Purification in a preferred basis] \label{lem_purif_almost_iid}
%Let $\rho_{C_1^n} \in \St(\cH_C^{\otimes n}, \tau_C^{\otimes n-r})$ be a $\binom{n}{r}$-\emph{almost-iid state} in $\tau_C $. 
Let $\ket{\Theta}_{AH} \in \cH_{AH}$ for some composite Hilbert space $\cH_{AH}$, and let $\cV_A \subseteq \cH_A$ be any subset of vectors such that $\,\mathrm{span}\,\cV_A$ admits an orthonormal basis $\big\{ \ket{\Psi_j}_A\big\}_{j \in \cT}$ with vectors $\ket{\Psi_j}_A \in \cV_A \ \  \forall\,j \in \cT$. 
Then, the following are equivalent:
\begin{enumerate}[(i)]
\item $\supp\big(\tr_H\big[\proj{\Theta}_{AH}\big]\big) \,\subseteq\, \mathrm{span}\, \cV_A$ \label{it_1}
\item $\ket{\Theta}_{AH} =\sum\limits_{j\in \cT} \alpha_j \ket{\Psi_j}_A \ket{\tilde{h}_j}_H$ for coefficients $ \alpha_j  \in \C$ and normalized vectors $\ket{\tilde{h}_j}_H \in \cH_H$ for all $j\in \cT$. \label{it_2}
\end{enumerate}
In addition, if $\ket{\Theta}_{AH}$ is normalized and satisfies $(ii)$, we have $1 = \braket{\Theta|\Theta}_{AH} = \sum_{j\in\cT} |\alpha_j|^2 $.
\end{lemma}
\begin{proof}
\underline{$\eqref{it_1} \implies \eqref{it_2}$}:
By the Schmidt decomposition, there exist orthonormal sets of vectors $\big\{ \ket{\chi_k}_A\big\}_{k} \subset \cH_A $ and $\big\{ \ket{\phi_k}_H\big\}_{k} \subset \cH_H $, and scalars $\gamma_k \geq 0$ such that $\ket{\Theta}_{AH} \,=\, \sum_k \gamma_k \ket{\chi_k}_A \ket{\phi_k}_H$.
Taking the partial trace, we obtain
\begin{align}
\tr_H\big[\proj{\Theta}_{AH}\big] \,=\, \sum_m \sum_{k,l} \gamma_k \gamma_l \ket{\chi_k}\bra{\chi_l}_A \braket{\phi_m|\phi_k}_H\braket{\phi_l|\phi_m}_H \,=\, \sum_k \gamma_k^2 \proj{\chi_k}_A \,.
\end{align}
By assumption~\eqref{it_1},
%and since the set $\big\{ \ket{\chi_k}_A\big\}_{k}$ is orthonormal, 
it follows that $ \ket{\chi_k} \in \mathrm{span}\, \cV_A$ for all $k$ such that $\gamma_k \neq 0$. (If, by contradiction, there were $ \ket{\chi_l} \notin \mathrm{span}\, \cV_A$ for some $l$ with $\gamma_l \neq 0$, then we would have $\tr_H\big[\proj{\Theta}_{AH}\big] \ket{\chi_l} = \gamma_l^2\ket{\chi_l}   \neq 0 \implies \supp(\dots) \not\subseteq  \mathrm{span}\, \cV_A $. $\contradiction$) Therefore, we may write $ \ket{\chi_k}_A = \sum_{j\in \cT} c^{(k)}_j \ket{\Psi_j}_A$ for some coefficients $c^{(k)}_j \in \C \ \forall\, j$ depending on $k$. With that, we find
\begin{align}
\ket{\Theta}_{AH} &= \sum_k \gamma_k \ket{\chi_k}_A \ket{\phi_k}_H  = \sum_k \gamma_k \sum_{j\in \cT} c^{(k)}_j \ket{\Psi_j}_A \ket{\phi_k}_H =  \sum_{j\in \cT}  \ket{\Psi_j}_A \Big( \sum_k  c^{(k)}_j \gamma_k \ket{\phi_k}_H \Big) \\
 &=  \sum_{j\in \cT}  \ket{\Psi_j}_A  \ket{{h}_j}_H = \sum_{j\in \cT}  \alpha_j \ket{\Psi_j}_A  \ket{\tilde{h}_j}_H \,,
\end{align}
where we have introduced coefficients $\alpha_j \geq 0$ to normalize the vector $ \ket{{h}_j}_H :=  \sum_k  c^{(k)}_j \gamma_k \ket{\phi_k}_H $ (with the convention that $\alpha_j = 0$ and \smash{$ \ket{\tilde{h}_j}_H $} are normalized but arbitrary if $ \ket{{h}_j}_H = 0$).

\underline{$\eqref{it_2} \implies \eqref{it_1}$}: For any orthonormal basis $\big\{ \ket{\phi_m}_H\big\}_{k}$ of $\cH_H$, we simply evaluate
\begin{align}
\tr_H\big[\proj{\Theta}_{AH}\big] \,&=\,  \sum_{i,j \in \cT} \alpha_i \alpha^*_j \sum_m \braket{\phi_m|\tilde{h}_i}_H\braket{\tilde{h}_j |\phi_m}_H \,\ket{\Psi_i}\bra{\Psi_j}_A  \, ,
\end{align}
which implies $\supp\big(\tr_H\big[\proj{\Theta}_{AH}\big]\big) \,\subseteq\, \mathrm{span}\, \cV_A$ since $\ket{\Psi_j}_A \in \cV_A \ \  \forall\,j \in \cT$. In addition, if $\ket{\Theta}_{AH}$ is normalized, we find $1 = \braket{\Theta|\Theta}_{AH} = \sum_{j\in\cT} |\alpha_j|^2 $ since the vectors $\big\{ \ket{\Psi_j}_A\big\}_{j \in \cT}$ are orthonormal and the vectors \smash{$\ket{\tilde{h}_j}_H$} are normalized for each $j \in \cT$.
\end{proof}

An important property of almost-iid states is that when tracing out many subsystems, we obtain a state that is close to an iid state. 
\begin{proposition} \label{prop_distance}
Let $n,r,s \in \N$ such that $r,s \leq n$, $\sigma_A \in \St(\cH_A)$, and $\rho_{A_1^n} \in \St^n(\cH_A, \sigma_A^{\otimes n-r})$.
Then
\begin{align} \label{eq_dist_claim}
\norm{ \tr_{n-s}[\rho_{A_1^n}]- \sigma_A^{\otimes s}}_1 
\leq 4 \sqrt{\frac{rs}{n}}\, .
\end{align}
%In particular, if $s=\lfloor (\log n)^{\frac{3}{2}} \rfloor$ and $r=O(\frac{n}{(\log n)^2})$ we have $\lim_{n \to \infty} \eps_n =0$.
\end{proposition}
\begin{proof}
Let $\tilde n:=\lfloor \frac{n}{s} \rfloor s \leq n$ and define $\rho_{A_1^{\tilde n}}:=\tr_{n-\tilde n}[\rho_{A_1^n}]$.
Let $\ket{\theta}_{AE}$ be a purification of $\sigma_A$.
By assumption $\rho_{A_1^n} \in \St^n(\cH_A, \sigma_A^{\otimes n-r})$ which according to~\cref{rmk_properties_almost_iid_states} implies $\rho_{A_1^{\tilde n}} \in \St^{\tilde n}(\cH_A, \sigma_A^{\otimes \tilde n-r})$.
Hence, due to the definition of this vector space, there exists a family $\{\ket{\Psi_t} \}_{t \in \cT}$ of orthonormal vectors from $\cV(\cH_{AE}^{\otimes \tilde n},\ket{\theta}_{AE}^{\otimes \tilde n-r})$ such that
\begin{align} \label{eq_ONB_dec}
\rho_{A_1^{\tilde n} E_1^{\tilde n}} = \sum_{t,t' \in \cT} \beta_{t,t'} \ket{ \Psi_t} \bra{\Psi_{t'}} \, ,
\end{align}
with $\sum_{t \in \cT } \beta_{t,t} =1$, where $\rho_{A_1^{\tilde n} E_1^{\tilde n}}$ is an extension of $\rho_{A_1^{\tilde n}}$.
Now, we may interpret $\ket{ \Psi_t}$ as a state on the $\lfloor \tfrac{n}{s} \rfloor$ blocks.
Since $\ket{ \Psi_t} \in \cV(\cH_{AE}^{\otimes \tilde n},\ket{\theta}_{AE}^{\otimes \tilde n-r})$, we know that $\ket{ \Psi_t}$ is at least in $\lfloor\tfrac{n}{s}\rfloor -r $ blocks of the form $\ket{\theta}^{\otimes s}$ and in the remaining $r$ blocks arbitrary.
For any $t \in \cT$, let $\cF_t$ denote the set of block indices $i \in \{1,\ldots,\lfloor \frac{n}{s} \rfloor\}$ on which  $\ket{\Psi_t}$ is not of the form $\ket{\theta}^{\otimes s}$ and note that 
\begin{align} \label{eq_size_error}
|\cF_t| \leq r \, .
\end{align}
Then, for any block specified by $i$, the total weight of the vectors $\ket{\Psi_t}$ that deviate from $\ket{\theta}^{\otimes s}$ for that block is given by
\begin{align} \label{eq_wi}
w_i := \sum_{t \in \cT} \beta_{t,t} \delta_{i \in \cF_t} \, .
\end{align}
Summing over all blocks yields
\begin{align} \label{eq_renato_dist}
\sum_{i=1}^{\lfloor \frac{n}{s} \rfloor} w_i 
\overset{\textnormal{\Cshref{eq_wi}}}{=} \sum_{i=1}^{\lfloor \frac{n}{s} \rfloor}\sum_{t \in \cT} \beta_{t,t} \delta_{i \in \cF_t}
= \sum_{t \in \cT} \beta_{t,t} \sum_{i=1}^{\lfloor\frac{n}{s}\rfloor}\delta_{i \in \cF_t}
= \sum_{t \in \cT} \beta_{t,t}  |\cF_t| 
\overset{\textnormal{\Cshref{eq_size_error}}}{\leq} \sum_{t \in \cT} \beta_{t,t}  r 
= r\, ,
\end{align}
where the final step uses $\sum_{t \in \cT } \beta_{t,t} =1$.
\Cref{eq_renato_dist} implies that 
\begin{align}\label{eq_renato_dist2}
w_j \leq \frac{r}{\lfloor \frac{n}{s} \rfloor} \quad \textnormal{for some} \quad j \in \left\{1,\ldots, \left\lfloor \frac{n}{s} \right\rfloor \right \} \, .
\end{align}
For a fixed $j \in \left\{1,\ldots, \left\lfloor \frac{n}{s} \right\rfloor \right \}$, let $\Pi$ be the projector onto the subspace spanned by $\{\ket{ \Psi_t} : j \not \in \cF_t \}$, i.e. 
\begin{align} \label{eq_def_PI}
\Pi = \sum_{t \in \cT \textnormal{ s.t. } j \not \in \cF_t } \proj{\Psi_t}   \qquad \textnormal{and} \qquad \Pi^\perp = \sum_{t \in \cT \textnormal{ s.t. } j \in \cF_t }\proj{\Psi_t} \, .
\end{align}
Note that $\Pi \, \Pi^\perp = 0$ and $\Pi + \Pi^\perp =\id_{\mathrm{span}\, \cV(\cH_{AE}^{\otimes \tilde n},\ket{\theta}_{AE}^{\otimes \tilde n-r})}$.
The operator \smash{$\rho'_{A_1^{\tilde n} E_1^{\tilde n}}:=\Pi \rho_{A_1^{\tilde n} E_1^{ \tilde n}} \Pi^\dagger $} is subnormalized, i.e., 
\begin{align} \label{eq_subnormalized_ds}
\tr[\rho'] = \tr[\Pi \rho \Pi^\dagger] \leq 1 \, ,
\end{align}
which can be seen, for example, via H\"older's inequality~\cite[Proposition~2.5]{Sutter_book}. 
The fidelity between two density operators is defined as \smash{$F(\rho,\sigma):=\norm{\sqrt{\rho}\sqrt{\sigma}}^2_1$}.
Hence, for any $\omega\geq 0$ we have
\begin{align} \label{eq_monotonicity_operator}
F(\omega,\rho') 
= \norm{\sqrt{\omega} \sqrt{\rho'}}^2_1 
=\tr[\rho'] \norm{\sqrt{\omega} \sqrt{\frac{\rho'}{\tr[\rho']}}}^2_1 
\overset{\textnormal{\Cshref{eq_subnormalized_ds}}}{\leq}  \norm{\sqrt{\omega} \sqrt{\frac{\rho'}{\tr[\rho']}}}^2_1 
= F\Big(\omega,\frac{\rho'}{\tr[\rho']}\Big)\, .
\end{align}
Using the permutation invariance of $\rho_{A_1^{\tilde n} E_1^{\tilde n}}$ we obtain
\begin{align} \label{eq_renato_dist3}
F(\rho_{A_1^s E_1^s}, \proj{\theta}^{\otimes s})
\overset{\textnormal{\Cshref{eq_monotonicity_operator}}}{\geq} F(\rho_{A_{(j-1)s+1}^{js} E_{(j-1)s+1}^{js}},\rho'_{A_{(j-1)s+1}^{js} E_{(j-1)s+1}^{js}}) 
\overset{\textnormal{DPI}}{\geq} F(\rho_{A_1^{\tilde n} E_1^{\tilde n}},\rho'_{A_1^{\tilde n} E_1^{\tilde n}}) \, ,
\end{align}
where the final step uses the data-processing inequality for the fidelity~\cite[Lemma B.4]{FR14}. By definition of the fidelity we have
\begin{align}
\sqrt{F(\rho_{A_1^{\tilde n} E_1^{\tilde n}},\rho'_{A_1^{\tilde n} E_1^{\tilde n}})}
=\norm{\sqrt{\rho} \sqrt{\rho'}}_1
=\tr \Big[ \sqrt{\sqrt{\rho} \rho' \sqrt{\rho} } \, \Big]
=\tr \Big[ \sqrt{\sqrt{\rho} \Pi \rho \Pi^\dagger \sqrt{\rho} } \, \Big]
=\tr[\Pi \rho ]\, . \label{eq_renato_dist_before4}
\end{align} 
By definition of $\Pi$ we have
\begin{align}
\tr[\Pi \rho ]
=1-\tr[\Pi^\perp \rho ]
\overset{\textnormal{\Cshref{eq_def_PI}}}{=} 1 - \sum_{t \in \cT, j \in \cF_t} \bra{\Psi_t} \rho \ket{\Psi_t}
\overset{\textnormal{\Cshref{eq_ONB_dec}}}{=} 1- \sum_{t \in \cT, j \in \cF_t} \beta_{t,t}
\overset{\textnormal{\Cshref{eq_wi}}}{=}1 - w_j \, .  \label{eq_renato_dist4}
\end{align}
The Fuchs-van der Graaf inequality~\cite{fuchs99} then gives
\begin{align}
\norm{\rho_{A_1^s E_1^s} - \proj{\theta}^{\otimes s}}_1
&\leq 2 \sqrt{1-F(\rho_{A_1^s E_1^s}, \proj{\theta}^{\otimes s})} \\
\overset{\textnormal{\Cshref{eq_renato_dist3,eq_renato_dist_before4,eq_renato_dist4}}}&{\leq} 2 \sqrt{1-(1 - w_j)^2}  \\
\overset{\textnormal{\Cshref{eq_renato_dist2}}}&{\leq} 2 \sqrt{2} \sqrt{\frac{r}{\lfloor \frac{n}{s} \rfloor}} \, . \label{eq_almostDONE-dist}
\end{align}
Recalling that the trace distance is contractive under the partial trace~\cite[Theorem 8.16]{wolf_notes} implies
\begin{align}
\norm{\rho_{A_1^s} - \sigma_A^{\otimes s}}_1  
\leq  \norm{\rho_{A_1^s E_1^s} - \proj{\theta}^{\otimes s}}_1 
\overset{\textnormal{\Cshref{eq_almostDONE-dist}}}{\leq} 2 \sqrt{2} \sqrt{\frac{r}{\lfloor \frac{n}{s} \rfloor}} \, . \label{eq_almostDONE2-dist}
\end{align}
To conclude the proof of the assertion, note that for $q=\lfloor \frac{n}{s} \rfloor$ we have
\begin{align}
sq = s \left \lfloor \frac{n}{s} \right \rfloor \leq n \leq s (q+1) \, . \label{eq_trivial_step}
\end{align}
Hence
\begin{align}
\frac{\sqrt{r/q}}{\sqrt{rs/n}}  
= \sqrt{\frac{n}{qs}}
\overset{\textnormal{\Cshref{eq_trivial_step}}}{\leq}  \sqrt{\frac{s(q+1)}{qs}}
=\sqrt{1+ \frac{1}{q}}
\leq \sqrt{2} \, , \label{eq_trivial_step2}
\end{align}
where the final step uses $q\geq 1$.
Putting everything together yields
\begin{align}
\norm{\rho_{A_1^s} - \sigma_A^{\otimes s}}_1  
\overset{\textnormal{\Cshref{eq_almostDONE2-dist}}}{\leq} 2 \sqrt{2} \sqrt{\frac{r}{\lfloor \frac{n}{s} \rfloor}}
\overset{\textnormal{\Cshref{eq_trivial_step2}}}{\leq} 4 \sqrt{\frac{rs}{n}} \, .
\end{align}
\end{proof}

Another property concerns the statistics of almost-iid and iid states when performing an iid measurement (as is used in a tomography procedure, for example)~\cite[Theorem 4.5.2]{renner_phd}. Let $\sigma \in \St(\cH)$ and assume that $n$ independent measurements with respect to a POVM $\cM = \{M_x\}_{x\in\cX}$ are performed on $n$ iid copies of $\sigma$, giving the outcomes ${\bf{x}} = (x_1,\dots,x_n)$ with $x_i \in \cX \ \forall\, i$. The outcomes ${\bf{x}}$ can be characterized by a \emph{frequency distribution} (or \emph{type}) $\lambda_{\bf{x}}$, which is defined as the following probability distribution on $\cX$,
\begin{align}
    \lambda_{\bf{x}}(y) := \frac{1}{n}\big|\{i:\,x_i = y\}\big|\,,
\end{align}
for all $y\in\cX$. Intuitively, $\lambda_{\bf{x}}(y)$ corresponds to the relative number of occurrences of $y$ in the sequence of outcomes ${\bf{x}} = (x_1,\dots,x_n)$. For large values of $n$, it follows by the law of large numbers that the frequency distribution $\lambda_{\bf{x}}$ is close to the probability distribution $P_X$ defined by $P_X(y) := \tr(M_y \sigma)$. 
\Cref{prop_statistics} shows that a similar statement holds when performing $n$ independent measurements on a $\binom{n}{r}$-almost-iid state in $\sigma$ for small enough values of $r$.
This proves that almost-iid and iid states cannot be distinguished by any practically feasible measurement. This is discussed in more detail in~\cite{MR_future}.

\begin{proposition}[Statistics of almost-iid states] \label{prop_statistics}
Let $n,r \in \N$ such that $r\leq \tfrac{1}{2}n$, $\sigma \in \St(\cH)$, and $\rho^{(n)} \in \St^{n}(\cH, \sigma^{\otimes n-r})$. Let $\cM = \{M_x\}_{x\in\cX}$ be a POVM on $\cH$, and $P_X(x) = \tr[M_x \sigma]$ for all $x\in \cX$. Then, for any $\eps > 0$, we have
\begin{align}
\underset{{\bf{x}}}{\PP}\big[ \norm{\lambda_{\bf{x}} - P_X }_1 >  f(\eps,r,n)\big] \leq \eps \,,
\end{align}
where the probability is taken over the outcomes ${\bf{x}} = (x_1,\dots,x_n)$ of the product measurement $\cM^{\otimes n}$ applied to $\rho^{(n)}$, and 
\begin{align}
    f(\eps,r,n) = 2\sqrt{\frac{\log{\big(\frac{1}{\eps}\big)}}{n} + \frac{|\cX|}{n}\log\Big( \frac{n}{2}+1\Big) + h\Big(\frac{r}{n}\Big) + \frac{2r}{n}\log(d) } + \frac{2r}{n}\,,
\end{align}
where $h(\cdot)$ is the binary entropy and $ d = \dim{\cH}$.
Furthermore, if $r=o(n)$, then $\lim_{n\to \infty} f(\eps,r,n) = 0$.
\end{proposition}
The proof is given in~\cref{app_proof_statistics}. 

%%%%%%%%%%%%%%%%%%%%%%%%%%%%%%%%%%%%%%%%%%%%%%%%%%%%%%%%%%%%%%%%%%%%
%%%%%%%%%%%%%%%%%%%%%%%%%%%%%%%%%%%%%%%%%%%%%%%%%%%%%%%%%%%%%%%%%%%%
\section{Quantum vs.~classical exponential de Finetti theorem} \label{sec_exp_deFinetti} 
In this section, we formally state the quantum exponential de Finetti theorem~\cite{renner07} which justifies the definition of almost-iid states. 
We then show that for classical almost-iid distributions, defined as convex combinations of product distributions with a small number of defects, no classical exponential de Finetti theorem can hold.
\begin{theorem-nocolor}[{Exponential de Finetti~\cite[Theorem~1]{renner07}}] \label{thm_exp_deFinetti_quantum}
Let $n,k,r \in \N$ and let $\cH$ be a $d$-dimensional Hilbert space. For any $\ket{\Phi^{(n+k)}} \in \mathrm{Sym}^{n+k}(\cH)$ there exists a probability measure $\nu$ on the unit sphere $\cB(\cH)$ and a family \smash{$\{\ket{\Psi^{(n)}_{r,\theta}}\}_{\theta}$} of states such that \smash{$\ket{\Psi^{(n)}_{r,\theta}} \in \mathrm{Sym}^n(\cH, \ket{\theta}^{\otimes n-r})$} and
\begin{align}
\norm{\tr_{k}[\proj{\Phi^{(n+k)}}] - \int \proj{\Psi^{(n)}_{r,\theta}}  \nu(\di \theta)}_1 \leq 3 k^d \ee^{-\frac{k(r+1)}{n+k}} \, .
\end{align}
\end{theorem-nocolor}

We next define almost-iid distributions which is the classical counterpart to almost-iid states. For a finite alphabet $\cX$, let $\mathrm{Sym}(\cX^n)$ denote the set of permutation-invariant distributions on $\cX^n$. Furthermore, for a distribution $q$ on $\cX$ let 
\begin{align} \label{eq_def-classical-V-set}
\cV(\cX^n,q^m):=\{ \pi ( q^m \times R_{X_1^{n-m}}) : \pi \in \cS_n, \, R_{X_1^{n-m}} \textnormal{ distribution on }\cX^{n-m} \} 
\end{align} 
and
\begin{align} \label{eq_def-almost-iid-distributions}
\mathrm{Sym}(\cX^{n}, q^{m}):= \mathrm{Sym}(\cX^{n}) \cap \mathrm{conv} \, \cV(\cX^{n}, q^{m}) \, .
\end{align}
Note that the set $\mathrm{Sym}(\cX^{n}, q^{m})$ is considerably simpler compared to $\mathrm{Sym}^n(\cH, \ket{\theta}^{\otimes m})$ since the former set consists of convex combinations of product distributions, whereas the latter set contains superpositions of product states. 

The following proposition states that no classical version of the quantum de Finetti theorem can exist where almost-iid states are replaced with almost-iid distributions.
\begin{proposition}[No classical exponential de Finetti result] \label{prop_no_classical_exp_dF}
Let $n\in \N$, $k=o(n)$, $r=o(n)$, $\cX$ be a finite alphabet of size $d$. There exists \smash{$P_{X_1^{n+k}} \in \mathrm{Sym}(\cX^{n+k})$} such that it is not possible to approximate $P_{X_1^n}$ with a probabilistic mixture of distributions \smash{$\{Q_{X_1^n}^{(q)} \}_{q}$} with \smash{$Q_{X_1^n}^{(q)} \in \mathrm{Sym}(\cX^{n}, q^{n-r})$} in the $\ell^1$-norm up to $\eps_d(n)$ such that $\lim_{n\to \infty} \eps_d(n)=0$.  
\end{proposition}
The proof of~\cref{prop_no_classical_exp_dF} is given in~\cref{app_no_classical_dF}.
As mentioned in~\cref{eq_deFinetti_classical}, if we are willing to choose $k$ large, more precisely $k=\omega(n)$, then a classical de Finetti theorem holds. However,~\cref{prop_no_classical_exp_dF} states that this is no longer the case for $k=o(n)$, even if the error term would decrease at a non-exponential rate.
Further note that~\cref{prop_no_classical_exp_dF} does not prohibit the existence of a classical exponential de Finetti theorem if the definition of almost-iid distributions,~\cref{eq_def-almost-iid-distributions}, is relaxed. For example, one could define the set $\cV(\cX^n,q^m)$ in~\cref{eq_def-classical-V-set} by only requiring the marginals of the distributions to yield an iid state instead of imposing a product structure between the iid part and the defects. However, such a structure would no longer be covered by the quantum version in~\cref{def_almost_product_state_mixed} and therefore, may potentially be difficult to work with.
% If we relaxed the definition of almost-iid distributions such that the defects could depend on the iid part, a classical exponential de Finetti theorem may be possible, however, at the cost that almost-iid distributions with correlations between the iid parts and the defects are potentially difficult to work with.
%\GM{Can we make more precise what we mean by "the defects could depend on the iid part" ? --> We could give up the product structure but only requiring that marginal on the good systems is iid. Structure would not be covered anymore by the quantum definition.}

One may still wonder what~\cref{thm_exp_deFinetti_quantum} produces if it is applied to a classical symmetric distribution. Can the resulting family of almost-iid states that approximates the classical distribution become classical as well? To analyze this, let \smash{$P_{X_1^{n+k}}\in \mathrm{Sym}(\cX^{n+k})$} be a symmetric classical distribution. We can apply~\Cref{thm_exp_deFinetti_quantum} to \smash{$P_{X_1^{n+k}}$} by embedding it into a permutation-invariant density matrix,
\begin{align}
\rho_{A_1^{n+k}}= \sum_{x^{n+k}_1} P_{X_1^{n+k}}(x_1^{n+k}) \proj{x_1} \cdots \proj{x_{n+k}} \, .
\end{align}
Let \smash{$\ket{\Phi^{(n+k)}}_{A_1^{n+k} E_1^{n+k}}$} be a symmetric purification of $\rho_{A_1^{n+k}}$. Consider the measurement channel on the $A$-system
\begin{align}
\cM: Y_{AE} \mapsto \sum_x \proj{x}_A Y_{AE} \proj{x}_A \, .
\end{align}
Because $\rho_{A_1^n}$ is classical, we have
\begin{align}
\rho_{A_1^n}
=\cM^{\otimes n}(\rho_{A_1^n}) 
=\cM^{\otimes n} \big(  \tr_{E_1^n}[ \rho_{A_1^n E_1^n} ]\big)
&= \tr_{E_1^n}\big[ \cM^{\otimes n} (\rho_{A_1^n E_1^n}) \big] \\
&=\tr_{E_1^n}\big[ \cM^{\otimes n} ( \tr_{k}[\proj{\Phi^{(n+k)}}]) \big] \, ,  \label{eq_cl_dF_part0}
\end{align}
where $ \rho_{A_1^n E_1^n} := \tr_{k}[\proj{\Phi^{(n+k)}}]$, and the penultimate step uses that $\cM$ only acts on the $A$-system and hence commutes with the partial trace over the $E$-system.
\Cref{thm_exp_deFinetti_quantum} shows that there exist a probability measure $\nu$ on the unit sphere $\cB(\cH)$ and, for each $\ket{\theta} \in \cB(\cH)$, a family \smash{$\{\ket{\Psi^{(n)}_{r,\theta}}\}_{\theta}$} of states such that \smash{$\ket{\Psi^{(n)}_{r,\theta}} \in \mathrm{Sym}^n(\cH, \ket{\theta}^{\otimes n-r})$}, and such that for $\sigma^{(n)}_{\theta} := \tr_{E_1^n}[ \cM^{\otimes n}( \proj{\Psi^{(n)}_{r,\theta}})]$ and $\bar \nu$ denoting the induced measure obtained by taking the partial trace, we have 
\begin{align}
\norm{\rho_{A_1^n}\! -\! \int\! \sigma^{(n)}_{\theta} \!  \bar \nu(\di \theta)   }_1
\!\!\overset{\textnormal{\Cshref{eq_cl_dF_part0}}}&{=}\!\!\norm{\tr_{E_1^n}\big[ \cM^{\otimes n} ( \tr_{k}[\proj{\Phi^{(n+k)}}]) \big] \!-\! \tr_{E_1^n}\Big[ \cM^{\otimes n} \Big(\! \int \proj{\Psi^{(n)}_{r,\theta}}  \nu(\di \theta)  \Big) \Big] }_1 \nonumber \\
&\leq \norm{ \tr_{k}[\proj{\Phi^{(n+k)}}]) - \int \proj{\Psi^{(n)}_{r,\theta}}  \nu(\di \theta)   }_1 \label{eq_cl_dF_part1} \\
\overset{\textnormal{\Cshref{thm_exp_deFinetti_quantum}}}&{\leq} 3 k^d \ee^{-\frac{k(r+1)}{n+k}} \, , \label{eq_cl_dF_part2}
\end{align}
where the second step uses the fact that the trace distance is contractive under trace-preserving completely positive maps~\cite[Theorem~8.16]{wolf_notes}. 
\Cref{prop_no_classical_exp_dF} now implies that the density operators $\sigma^{(n)}_{\theta}$ cannot be given almost-iid distributions in the sense of~\cref{eq_def-almost-iid-distributions}. Intuitively, this happens since the states $\ket{\theta}$ appearing in~\cref{eq_cl_dF_part1} are not necessarily classical and therefore, the almost-iid structure is not preserved after the measurement.

%are in general not classical distributions but quantum mechanical objects. 
%\GM{But it could be an object that is still classical but not of the form of Eq. (38) ? Since we measured it, it should actually be classical, but it doesn't have an almost-iid form anymore, I guess ... Open question: maybe there is another good classical relaxation. Note: Thetas are in general not classical, this is why the structure is not maintained after measurement.}

%%%%%%%%%%%%%%%%%%%%%%%%%%%%%%%%%%%%%%%%%%%%%%%%%%%%%%%%%%%%%%%%%%%%
%%%%%%%%%%%%%%%%%%%%%%%%%%%%%%%%%%%%%%%%%%%%%%%%%%%%%%%%%%%%%%%%%%%%
\section{Conditional entropy of almost-iid states}  \label{sec_cond_entropy}
In this section, we prove that the conditional entropy of almost-iid states asymptotically coincides with the conditional entropy of iid states. 
In the proof, we develop technical tools that may be of independent interest. 

%For $\rho \in \St(\cH)$ and $\sigma \geq 0$ the relative entropy is given by $D(\rho \| \sigma):=\tr[\rho (\log \rho - \log \sigma)]$ if the support of $\rho$ is included in the support of $\sigma$, and $+\infty$ otherwise.
For $\rho_{AB} \in \St(\cH_{AB})$ the conditional entropy is defined as $H(A|B)_{\rho}= H(AB)_{\rho} - H(B)_{\rho}$, where $H(B)_{\rho} := - \tr[\rho_B \log \rho_B]$ is the von Neumann entropy.
It is straightforward to see that the conditional entropy is additive for iid states, i.e.,
\begin{align}
 \frac{1}{n} H(A_1^n|B_1^n)_{\rho^{\otimes n}} = H(A|B)_{\rho} \, .
\end{align}
We next show that this property is preserved for almost-iid states in the limit $n \to \infty$.
\begin{theorem} \label{thm_entropy_of_almost_iid}
Let $\sigma_{AB} \in \St(\cH_{AB})$ and $\rho_{A_1^n B_1^n} \in \St^n(\cH_{AB}, \sigma_{AB}^{\otimes n-r})$ for $r=o(n)$. Then 
\begin{align} \label{eq_almostProduct_entropy_cor}
 \frac{1}{n} H(A_1^n|B_1^n)_{\rho} = H(A|B)_{\sigma} + \frac{o(n)}{n} \, .
\end{align}
\end{theorem}
The main difficulty in proving the assertion of~\cref{thm_entropy_of_almost_iid} is the fact that almost-iid states are more general than just convex mixtures of product states where each element in the convex sum has a certain number of defects (see~\cref{ex_almost_iid}). A look at~\cref{def_almost_product_state_mixed} reveals that almost-iid states may contain superpositions which store long range correlations and entanglement. Dealing with these superpositions is the main technical challenge in the proof. We do this utilizing tools from one-shot information theory such as R\'enyi and smooth entropies.
We also want to emphasize that the superpositions in the definition of almost-iid states are crucial for making the exponential de Finetti theorem (\cref{thm_exp_deFinetti_quantum}) possible. As shown in~\cref{prop_no_classical_exp_dF}, no exponential de Finetti theorem can exist without superpositions.

Before presenting the proof of~\cref{thm_entropy_of_almost_iid}, which is given in~\cref{sec_pf_main_thm}, we need to define some entropic quantities. 
We note that an alternative proof using different techniques that may therefore be of independent interest is given in~\cref{app_alternative_proof}.

\subsection{Entropic quantities}
For $\alpha \in [1/2,1) \cup (1,\infty)$ the \emph{sandwiched R\'enyi divergence}~\cite{MLDSFT13,WWY14} is given by
\begin{align} \label{eq_sandwhiched}
D_{\alpha}(\rho \| \sigma) :=\frac{1}{\alpha -1 } \log \tr\big[(\sigma^{\frac{1-\alpha}{2\alpha}} \rho \, \sigma^{\frac{1-\alpha}{2\alpha}})^{\alpha}\big] \, .
\end{align}
For $\alpha = \frac{1}{2}$ we have \smash{$D_{\frac{1}{2}}(\rho \| \sigma)=- \log F(\rho,\sigma)$}.
In the limits $\alpha \to 1$ and $\alpha \to \infty$ the sandwiched R\'enyi divergence converges to the relative entropy $D(\rho \| \sigma)$ and the max-relative entropy~\cite{renner_phd,datta09}
\begin{align}
D_{\max}(\rho \| \sigma) :=\inf\{\lambda \in \R: \rho \leq 2^\lambda \sigma \} \, ,
\end{align}
respectively.
For $\rho_{AB} \in \St(\cH_{AB})$, the R\'enyi divergence can be used to define a conditional R\'enyi entropy~\cite{marco_book}
\begin{align} \label{eq_def_cond_renyi_entropy}
H_{\alpha}(A|B)_\rho := - \min_{\sigma_B \in \St(\cH_B)} D_{\alpha}(\rho_{AB} \| \id_A \otimes \sigma_B) \, , 
\end{align}
which converges to the conditional entropy $H(A|B)_{\rho}$ for $\alpha \to 1$ and to the conditional min-entropy $H_{\min}(A|B)_{\rho}$ for $\alpha \to \infty$. For $\alpha = 1/2$ we obtain the conditional max-entropy $H_{\max}(A|B)_{\rho}$.
The trace distance between two states $\rho,\sigma \in \St(\cH)$ is given by $\Delta(\rho,\sigma):=\frac{1}{2}\norm{\rho-\sigma}_1$ and the purified distance~\cite{marco_book} is defined as \smash{$P(\rho,\sigma):=\sqrt{1-F(\rho,\sigma)}$}. The Fuchs-van der Graaf inequality~\cite{fuchs99} implies $P(\rho,\sigma) \geq \Delta(\rho,\sigma)$.
For $\rho \in \St(\cH)$ and $\eps\in(0,1)$ define the $\eps$-ball around $\rho$ by $\cB_{\eps}(\rho):=\{\rho' \in \St(\cH): P(\rho,\rho') \leq \eps \}$.
We then define a smooth variant of the min- and max-entropy by
\begin{align}
H^\eps_{\min}(A|B)_{\rho} :=\max_{\rho' \in \cB_{\eps}(\rho)} H_{\min}(A|B)_{\rho'} \qquad \textnormal{and} \qquad  H^\eps_{\max}(A|B)_{\rho} :=\min_{\rho' \in \cB_{\eps}(\rho)} H_{\max}(A|B)_{\rho'} \, .
\end{align}
% Similarly, we define the smooth max-relative entropy
% \begin{align}
%  D_{\max}^\eps(\rho \| \sigma ) := \min_{\rho' \in \cB_{\eps}(\rho)}  D_{\max}(\rho' \| \sigma ) \, .
% \end{align}

\subsection{Asymptotic equipartition property for almost-iid states}
Another statement which can be proven using similar techniques and may be of independent interest is a strong asymptotic equipartition property (AEP) for almost-iid states. 
To understand this, let us recall the AEP for iid states~\cite{tomamichel09,marco_book}. This fundamental result ensures that for any density operator $\rho_{AB} \in \St(\cH_{AB})$ and any $\eps \in (0,1)$ we have
\begin{align} \label{eq_iid_AEP}
 \frac{1}{n} H^\eps_{\min}(A_1^n|B_1^n)_{\rho^{{\otimes n}}} 
= H(A|B)_{\rho} + \frac{o(n)}{n} \qquad \textnormal{and} \qquad
\frac{1}{n} H^\eps_{\max}(A_1^n|B_1^n)_{\rho^{{\otimes n}}} = H(A|B)_{\rho} + \frac{o(n)}{n} \, .
\end{align}
Note that~\cref{eq_iid_AEP} is called a \emph{strong} AEP as the error term $\frac{o(n)}{n}$ vanishes in the limit $n \to \infty$ for any fixed $\eps \in (0,1)$. Furthermore, it is understood how fast the error term vanishes for finite values of $n$~\cite{TH13}. 
We show that~\cref{eq_iid_AEP} remains valid when replacing iid states with almost-iid states.
\begin{proposition}[Strong AEP for almost-iid states] \label{prop_strong_AEP}
Let $\eps \in (0,1)$, $\sigma_{AB} \in \St(\cH_{AB})$, and $\rho_{A_1^n B_1^n} \in \St^n(\cH_{AB}, \sigma_{AB}^{\otimes n-r})$ for $r=o(n)$. Then 
\begin{align} \label{eq_strong_AEP}
 \frac{1}{n} H^\eps_{\min}(A_1^n|B_1^n)_{\rho} 
= H(A|B)_{\sigma} + \frac{o(n)}{n} \qquad \textnormal{and} \qquad
\frac{1}{n} H^\eps_{\max}(A_1^n|B_1^n)_{\rho} = H(A|B)_{\sigma} + \frac{o(n)}{n} \, .
\end{align}
\end{proposition}

In~\cite[Theorem~4.4.1]{renner_phd} it was shown that the conditional smooth min-entropy for almost-iid states, which are classical on one subsystem, asymptotically coincides with the conditional entropy of iid states. This was crucial to prove security of quantum key distribution via a de Finetti argument.
Using the duality of conditional entropy, the result can be lifted to the smooth max-entropy of almost-iid states~\cite{YC22}.
In addition, for the case of pure almost-iid states, a similar result has been proven based on~\cite{renner_phd} in~\cite[Lemma 11]{YC22}. Here, the presented proof of~\cref{prop_strong_AEP} is more general as it applies for mixed almost-iid states and is also more modular allowing one to distill other results such as~\cref{thm_entropy_of_almost_iid}. 
Beyond these results, to the best of our knowledge, little is known about how entropic functions behave for almost-iid states. 

%%%%%%%%%%%%%%%%%%%%%%%%%%%%%%%%%%%%%%%%%%%%%%%%%%%%%%%%%%%%%%%%%%%%%%%%%
\subsection{Proof of~\cref{prop_strong_AEP}} \label{sec_proof_AEP}
Let $\rho_{A_1^n B_1^n} \in \St^n(\cH_{AB}, \sigma_{AB}^{\otimes n-r})$. By definition, there exist a purification $\ket{\theta}_{ABE}$ of $\sigma_{AB}$ and an extension $\rho_{A_1^n B_1^n E_1^n}$ of $\rho_{A_1^n B_1^n}$ that can be written as
\begin{align} \label{eq_decomp_psi}
\rho_{A_1^n B_1^n E_1^n} = \sum_{t,t' \in \cT} \beta_{t,t'} \ket{\Psi_t} \bra{\Psi_{t'}}_{A_1^n B_1^n E_1^n} \, ,
\end{align}
for a family $\{\ket{\Psi_t}_{A_1^n B_1^n E_1^n}\}_{t \in \cT}$ of orthonormal vectors from $\cV(\cH_{ABE}^{\otimes n}, \ket{\theta}_{ABE}^{\otimes n-r})$
with $\beta_{t,t'} \in \C$ satisfying $\sum_{t \in \cT} \beta_{t,t} =1$ and 
\begin{align} \label{eq_size_T}
\log |\cT| \leq n \, h\Big( \frac{r}{n} \Big) + r \log d_{ABE} \, .
\end{align}
Let 
\begin{align} \label{eq_tilde_rho}
\tilde \rho_{A_1^n B_1^n E_1^n T} := \sum_{t \in \cT} \beta_{t,t} \underbrace{\proj{\Psi_t}_{A_1^n B_1^n E_1^n}}_{=:\tilde \rho^{(t)}} \otimes \proj{t}_T \, .
\end{align}
\begin{lemma} \label{lem_pinching_statement}
For the setting defined above, we have
\begin{align} \label{eq_RR_operator_ineq}
\rho_{A_1^n B_1^n E_1^n} \leq |\cT| \tilde \rho_{A_1^n B_1^n E_1^n}\,, \quad \textnormal{and hence} \quad\rho_{A_1^n B_1^n} \leq |\cT| \tilde \rho_{A_1^n B_1^n} \, .
\end{align}
\end{lemma}
\begin{proof}
The proof idea is similar to~\cite[Proof of Lemma~3.1.13]{renner_phd} but is based on pinching maps and therefore works for a more general setup.
Consider the pinching map
\begin{align}
\cP: Z \mapsto \sum_{t \in \cT} \proj{\Psi_t} Z \proj{\Psi_t} \, .
\end{align}
Using the fact that $\{\ket{\Psi_t} \}_{t \in \cT}$ is orthonormal, we have
\begin{align}
\cP(\rho_{A_1^n B_1^n E_1^n}) 
\overset{\textnormal{\Cshref{eq_decomp_psi}}}{=} \sum_{t,k,\ell \in \cT} \beta_{k,\ell} \proj{\Psi_t} \ket{\Psi_k} \bra{\Psi_{\ell}} \proj{\Psi_t}
= \sum_{t \in \cT} \beta_{t,t} \proj{\Psi_t}
\overset{\textnormal{\Cshref{eq_tilde_rho}}}{=}\tilde \rho_{A_1^n B_1^n E_1^n} \, . \label{eq_pinching_ds1}
\end{align}
Hence, we find
\begin{align}
\tilde \rho_{A_1^n B_1^n E_1^n}
\overset{\textnormal{\Cshref{eq_pinching_ds1}}}{=} \cP(\rho_{A_1^n B_1^n E_1^n}) 
\overset{\textnormal{pinching inequality}}{\geq} \frac{1}{|\cT|} \rho_{A_1^n B_1^n E_1^n} \, . \label{eq_piniching_ds2}
\end{align}
The interested reader can find more information on pinching maps, including a proof of the pinching inequality in~\cite[Lemma~3.5]{Sutter_book}. Since the partial trace is a completely positive map~\cref{eq_piniching_ds2} implies
\begin{align}
\tilde \rho_{A_1^n B_1^n} \geq \frac{1}{|\cT|} \rho_{A_1^n B_1^n} \, .
\end{align}
\end{proof}
\begin{lemma} \label{lem_key_alpha}
Let $n,r \in \N$ such that $r \leq n$, $\sigma_{AB} \in \St(\cH_{AB})$, $\rho_{A_1^n B_1^n } \in \St^n(\cH_{AB}, \sigma_{AB}^{\otimes n-r})$, $d_{AB} = \dim \cH_{AB} $, and $d_A = \dim \cH_A$. Then 
\begin{align} \label{eq_large_alpha}
\frac{1}{n} H_{\alpha}(A_1^n|B_1^n)_{\rho} \geq H_{\alpha}(A|B)_{\sigma} - \frac{2 r}{n} \log d_A - \frac{\alpha}{\alpha -1} \left( h\Big(\frac{r}{n}\Big) + \frac{2r}{n} \log d_{AB} \right)  \quad \forall \alpha >1 \, 
\end{align}
and
\begin{align} \label{eq_small_alpha}
\frac{1}{n} H_{\alpha}(A_1^n|B_1^n)_{\rho} \leq H_{\alpha}(A|B)_{\sigma} + \frac{2 r}{n} \log d_A + \frac{1}{1- \alpha } \left( h\Big(\frac{r}{n}\Big) + \frac{2r}{n} \log d_{AB} \right) \quad \forall \alpha \in [\tfrac{1}{2},1) \, .
\end{align}
\end{lemma}
\begin{proof}
We start by proving~\cref{eq_large_alpha}. For $\alpha >1$ and $\rho$, $\tilde \rho$ defined above, we have
\begin{align}
\frac{\alpha}{1-\alpha} \log |\cT| + H_{\alpha}(A_1^n|B_1^n)_{\tilde \rho}
\overset{\textnormal{\Cshref{eq_sandwhiched,eq_def_cond_renyi_entropy}}}&{=} \max_{\sigma \in \St(\cH_{B}^{\otimes n})} \frac{1}{1-\alpha} \log \tr\big[(\sigma^{\frac{1-\alpha}{2\alpha}} \tilde \rho_{A_1^n B_1^n} |\cT| \sigma^{\frac{1-\alpha}{2\alpha}})^\alpha\big] \\
\overset{\textnormal{\Cshref{lem_pinching_statement}}}&{\leq} \max_{\sigma \in \St(\cH_{B}^{\otimes n})} \frac{1}{1-\alpha} \log \tr\big[(\sigma^{\frac{1-\alpha}{2\alpha}} \rho_{A_1^n B_1^n} \sigma^{\frac{1-\alpha}{2\alpha}})^\alpha\big] \\
\overset{\textnormal{\Cshref{eq_sandwhiched,eq_def_cond_renyi_entropy}}}&{=} H_{\alpha}(A_1^n|B_1^n)_{\rho} \, ,
\end{align}
where the inequality step used that the function $X \mapsto \tr[X^{\alpha}]$ is monotone~\cite[Theorem~2.10]{carlen_book}.
Furthermore, we have
\begin{align}
 H_{\alpha}(A_1^n|B_1^n)_{\tilde \rho}
 \overset{\textnormal{\cite[Eq.~5.41]{marco_book}}}&{\geq}  H_{\alpha}(A_1^n|B_1^n T)_{\tilde \rho}\\
 \overset{\textnormal{\cite[Prop.~5.4]{marco_book}}}&{=} \frac{\alpha}{1-\alpha} \log \left( \sum_{t \in \cT} \beta_{t,t} \exp \Big( \frac{1-\alpha}{\alpha} H_{\alpha}(A_1^n|B_1^n)_{\tilde \rho^{(t)}} \Big) \right) \\
 &\geq \min_{t \in \cT} H_{\alpha}(A_1^n|B_1^n)_{\tilde \rho^{(t)}} \, ,
\end{align}
where the final step uses that the logarithm is a quasi-linear function and that $\beta_{t,t} \in [0,1]$ with $\sum_{t \in \cT} \beta_{t,t} =1$.
Recalling that $\tilde \rho^{(t)} = \proj{\Psi_t}$ for $\ket{\Psi_t} \in \cV(\cH_{ABE}^{\otimes n}, \ket{\theta}_{ABE}^{\otimes n-r})$ and using the additivity of the R\'enyi entropies under tensor products allows us to write for any $t \in \cT$
\begin{align}
H_{\alpha}(A_1^n|B_1^n)_{\tilde \rho^{(t)}} 
=(n-r) H_{\alpha}(A|B)_{\sigma} + H_{\alpha}(A_1^r |B_1^r)_{\Omega} 
\overset{\textnormal{\cite[Lem.~5.11]{marco_book}}}{\geq} (n-r) H_{\alpha}(A|B)_{\sigma} - r \log d_A \, .
\end{align}
Putting everything together yields
\begin{align}
\frac{1}{n} H_{\alpha}(A_1^n|B_1^n)_{\rho} 
&\geq \frac{n-r}{n} H_{\alpha}(A|B)_{\sigma} - \frac{r}{n} \log d_A - \frac{1}{n} \frac{\alpha}{\alpha-1} \log |\cT| \\
 \overset{\textnormal{\cite[Lem.~5.11]{marco_book}\&\Cshref{eq_size_T}}}&{\geq} H_{\alpha}(A|B)_{\sigma} - \frac{2 r}{n} \log d_A - \frac{\alpha}{\alpha -1} \left( h\Big(\frac{r}{n}\Big) + \frac{2r}{n} \log d_{AB} \right) \, ,
\end{align}
where in the final step, we used $d_{ABE} = d_{AB}^2$. This proves~\cref{eq_large_alpha}.

The statement from~\cref{eq_small_alpha} follows similarly.
For $\alpha \in [1/2,1)$ we have
\begin{align}
\frac{\alpha}{1-\alpha} \log |\cT| + H_{\alpha}(A_1^n|B_1^n)_{\tilde \rho}
\overset{\textnormal{\Cshref{eq_sandwhiched,eq_def_cond_renyi_entropy}}}&{=} \max_{\sigma \in \St(\cH_{B}^{\otimes n})} \frac{1}{1-\alpha} \log \tr\big[(\sigma^{\frac{1-\alpha}{2\alpha}} \tilde \rho_{A_1^n B_1^n} |\cT| \sigma^{\frac{1-\alpha}{2\alpha}})^\alpha\big] \\
\overset{\textnormal{\Cshref{lem_pinching_statement}}}&{\geq} \max_{\sigma \in \St(\cH_{B}^{\otimes n})} \frac{1}{1-\alpha} \log \tr\big[(\sigma^{\frac{1-\alpha}{2\alpha}} \rho_{A_1^n B_1^n} \sigma^{\frac{1-\alpha}{2\alpha}})^\alpha\big] \\
\overset{\textnormal{\Cshref{eq_sandwhiched,eq_def_cond_renyi_entropy}}}&{=} H_{\alpha}(A_1^n|B_1^n)_{\rho} \, , \label{eq_new_nonduality1}
\end{align}
where the inequality step used that the function $X \mapsto \tr[X^{\alpha}]$ is monotone~\cite[Theorem~2.10]{carlen_book}.
In addition, we have
\begin{align}
 H_{\alpha}(A_1^n|B_1^n)_{\tilde \rho}
 \overset{\textnormal{\cite[Eq.~5.96]{marco_book}}}&{\leq}  H_{\alpha}(A_1^n|B_1^n T)_{\tilde \rho} + \log|\cT| \\
 \overset{\textnormal{\cite[Prop.~5.4]{marco_book}}}&{=} \frac{\alpha}{1-\alpha} \log \left( \sum_{t \in \cT} \beta_{t,t} \exp \Big( \frac{1-\alpha}{\alpha} H_{\alpha}(A_1^n|B_1^n)_{\tilde \rho^{(t)}} \Big) \right) + \log|\cT| \\
 &\leq \max_{t \in \cT} H_{\alpha}(A_1^n|B_1^n)_{\tilde \rho^{(t)}}+ \log|\cT| \, , \label{eq_new_nonduality2}
\end{align}
where the final step uses that the logarithm is a quasi-linear function and that $\beta_{t,t} \in [0,1]$ with $\sum_{t \in \cT} \beta_{t,t} =1$.
Since $\tilde \rho^{(t)} = \proj{\Psi_t}$ for $\ket{\Psi_t} \in \cV(\cH_{ABE}^{\otimes n}, \ket{\theta}_{ABE}^{\otimes n-r})$, the additivity of the R\'enyi entropies under tensor products implies for any $t \in \cT$
\begin{align}
H_{\alpha}(A_1^n|B_1^n)_{\tilde \rho^{(t)}} 
=(n-r) H_{\alpha}(A|B)_{\sigma} + H_{\alpha}(A_1^r |B_1^r)_{\Omega} 
\overset{\textnormal{\cite[Lem.~5.11]{marco_book}}}{\leq} (n-r) H_{\alpha}(A|B)_{\sigma} + r \log d_A \, . \label{eq_new_nonduality3}
\end{align}
Combining~\cref{eq_new_nonduality1,eq_new_nonduality2,eq_new_nonduality3} yields
\begin{align}
\frac{1}{n} H_{\alpha}(A_1^n|B_1^n)_{\rho}
&\leq  \frac{n-r}{n} H_{\alpha}(A|B)_{\sigma} + \frac{r}{n} \log d_A  + \frac{1}{1-\alpha} \frac{1}{n} \log |\cT|\\
\overset{\textnormal{\cite[Lem.~5.11]{marco_book}\&\Cshref{eq_size_T}}}&{\leq} H_{\alpha}(A|B)_{\sigma} + \frac{2r}{n} \log d_A + \frac{1}{1-\alpha}\left( h\Big(\frac{r}{n}\Big) + \frac{2r}{n} \log d_{AB} \right) \, ,
\end{align}
where in the final step we used that $d_E = d_{AB}$.
\end{proof}

We are now equipped with all the tools we need to prove the assertion of~\cref{prop_strong_AEP}. This will be done in four steps, by proving two inequalities (direct and converse part) for both the smooth min- and max-entropy.
\begin{enumerate}[(i)]
\item Direct part for smooth min-entropy: For $\alpha >1$ consider the error term
\begin{align}
\delta_{\eps}(n,r,\alpha):= \frac{2 r}{n} \log d_A + \frac{\alpha}{\alpha -1} \left( h\Big(\frac{r}{n}\Big) + \frac{2r}{n} \log d_{AB} \right) +  \frac{1}{n} \frac{1}{\alpha-1} \log \frac{1}{\eps^2} + \frac{1}{n}\log\frac{1}{1-\eps^2} \, .
\end{align}
We can write
\begin{align}
\frac{1}{n} H_{\min}^{\eps}(A_1^n | B_1^n)_{\rho}
\overset{\textnormal{\cite[Prop.~2.2]{buscemi19}}}&{\geq} \frac{1}{n} H_{\alpha}(A_1^n | B_1^n)_{\rho} - \frac{1}{n} \frac{1}{\alpha-1} \log \frac{1}{\eps^2} - \frac{1}{n}\log\frac{1}{1-\eps^2} \\
\overset{\textnormal{\Cshref{lem_key_alpha}}}&{\geq} H_{\alpha}(A | B)_{\sigma} - \delta_{\eps}(n,r,\alpha) \\
\overset{\textnormal{\cite[Lem.~8]{tomamichel09}}}&{\geq} H(A | B)_{\sigma} - \delta_{\eps}(n,r,\alpha) - 4(\alpha -1) (\log \eta)^2 \, ,
\end{align}
for a constant $\eta=\sqrt{2^{-H_{\min}(A|B)_{\sigma}}} + \sqrt{2^{H_{\max}(A|B)_{\sigma}}} +1$.
For a choice $\alpha = 1 + 1/\log(r/n)$ and recalling that $r=o(n)$ we see that
\begin{align}
\delta_{\eps}(n,r,\alpha)
=\frac{o(n)}{n}\, ,
\end{align}
where we used that $\lim_{x\to 0} (\log(x)+1)h(x) = 0$ and $\lim_{x\to 0} (\log(x)+1) x = 0$ . Hence, we obtain
\begin{align} \label{eq_direct_Hmin}
\frac{1}{n} H_{\min}^{\eps}(A_1^n | B_1^n)_{\rho} \geq H(A | B)_{\sigma} - \frac{o(n)}{n} \, .
\end{align}

\item Direct part for smooth max-entropy: For $\alpha \in [1/2,1)$ consider the error term
\begin{align}
\delta'_{\eps}(n,r,\alpha):= \frac{2 r}{n} \log d_A + \frac{1}{1-\alpha} \left( h\Big(\frac{r}{n}\Big) + \frac{2r}{n} \log d_{AB} \right) +  \frac{1}{n} \frac{\alpha}{1-\alpha} \log \frac{1}{\eps}  \, .
\end{align}
We can write
\begin{align}
\frac{1}{n} H_{\max}^{\eps}(A_1^n | B_1^n)_{\rho}
\overset{\textnormal{\cite[Prop.~2.2]{buscemi19}}}&{\leq} \frac{1}{n} H_{\alpha}(A_1^n | B_1^n)_{\rho} + \frac{1}{n} \frac{\alpha}{1-\alpha} \log \frac{1}{\eps}  \\
\overset{\textnormal{\Cshref{lem_key_alpha}}}&{\leq} H_{\alpha}(A | B)_{\sigma} + \delta'_{\eps}(n,r,\alpha) \\
\overset{\textnormal{\cite[Lem.~8]{tomamichel09} \& \cite[Prop.~5.7]{marco_book}}}&{\leq} H(A | B)_{\sigma} + \delta'_{\eps}(n,r,\alpha) + 4(1-\alpha) (\log \eta)^2 \, .
\end{align}
Similarly as above, choosing $\alpha = 1 - 1/\log(r/n)$ yields $\delta'_{\eps}(n,r,\alpha) = \frac{o(n)}{n}$ and hence
\begin{align}\label{eq_direct_Hmax}
\frac{1}{n} H_{\max}^{\eps}(A_1^n | B_1^n)_{\rho} \leq H(A | B)_{\sigma} + \frac{o(n)}{n} \, .
\end{align}

\item Converse part for smooth min-entropy: For a fixed $\eps \in (0,1)$ consider an arbitrary $\eps' \in (0,1-\eps)$. Then
\begin{align} \label{eq_converse_Hmin}
\frac{1}{n} H_{\min}^{\eps}(A_1^n | B_1^n)_{\rho} \!
\overset{\textnormal{\cite[Eq.~6.107]{marco_book}}}{\leq}\! \frac{1}{n} H_{\max}^{\eps'}(A_1^n | B_1^n)_{\rho} \!+\! \frac{1}{n}  \log \frac{1}{1\!-\!(\eps \!+\! \eps')^2}
\overset{\textnormal{\Cshref{eq_direct_Hmax}}}{\leq} H(A | B)_{\sigma} \!+\! \frac{o(n)}{n} \, .
\end{align}

\item Converse part for smooth max-entropy: For a fixed $\eps \in (0,1)$ consider an arbitrary $\eps' \in (0,1-\eps)$. Then
\begin{align} \label{eq_converse_Hmax}
\frac{1}{n} H_{\max}^{\eps'}(A_1^n | B_1^n)_{\rho}\!
\overset{\textnormal{\cite[Eq.~6.107]{marco_book}}}{\geq}\! \frac{1}{n} H_{\min}^{\eps}(A_1^n | B_1^n)_{\rho} \!-\! \frac{1}{n}  \log \frac{1}{1\!-\!(\eps \!+\! \eps')^2}
\!\overset{\textnormal{\Cshref{eq_direct_Hmin}}}{\geq}\! H(A | B)_{\sigma} \!-\! \frac{o(n)}{n} \, .
\end{align}
\end{enumerate}

\noindent Combining~\cref{eq_direct_Hmin,eq_direct_Hmax,eq_converse_Hmin,eq_converse_Hmax} completes the proof.\qed
%%%%%%%%%%%%%%%%%%%%%%%%%%%%%%%%%%%%%%%%%%%%%%%%%%%%%%%%%%%%%%%%%%%%
\subsection{Proof of~\cref{thm_entropy_of_almost_iid}} \label{sec_pf_main_thm}
The monotonicity of the R\'enyi divergence in $\alpha$~\cite{MLDSFT13} implies that for $\alpha>1$ we have
\begin{align}
\frac{1}{n}H(A_1^n | B_1^n)_{\rho}
&\geq \frac{1}{n}H_{\alpha}(A_1^n | B_1^n)_{\rho}\\
\overset{\textnormal{\Cshref{lem_key_alpha}}}&{\geq} H_{\alpha}(A|B)_{\sigma} - \frac{2r}{n} \log d_A - \frac{\alpha}{\alpha -1} \left( h\Big(\frac{r}{n}\Big) + \frac{2r}{n} \log d_{AB} \right) \\
\overset{\textnormal{\cite[Lem.~8]{tomamichel09}}}&{\geq} H(A|B)_{\sigma} \!-\! \frac{2r}{n} \log d_A \!-\! \frac{\alpha}{\alpha -1}\left( h\Big(\frac{r}{n}\Big) \!+\! \frac{2r}{n} \log d_{AB} \right) \!-\! 4(\alpha \!-\!1) (\log \eta)^2 \, ,
\end{align}
for a constant $\eta=\sqrt{2^{-H_{\min}(A|B)_{\sigma}}} + \sqrt{2^{H_{\max}(A|B)_{\sigma}}} +1$.
Choosing $\alpha = 1 + 1/\log(r/n)$ and recalling that $r=o(n)$ yields\footnote{Note that $\lim_{x\to 0} (\log(x)+1)h(x) = 0$ and $\lim_{x\to 0} (\log(x)+1)x = 0$.}
\begin{align} \label{eq_part1_entropy}
\frac{1}{n}H(A_1^n | B_1^n)_{\rho}
\geq H(A|B)_{\sigma} - \frac{o(n)}{n} \, .
\end{align}
To see the other direction, note that for $\eps_n = 2\sqrt{\frac{r}{n}}$ with $r=o(n)$
\begin{align}
\frac{1}{n} H(A_1^n|B_1^n)_{\rho}
\overset{\textnormal{chain rule}}&{=} \frac{1}{n} \sum_{i=1}^n H(A_i|A_1^{i-1}B_1^n)_{\rho}\\
\overset{\textnormal{SSA~\cite{LieRus73_1}}}&{\leq} \frac{1}{n} \sum_{i=1}^n H(A_i| B_i)_{\rho}\\
\overset{\textnormal{perm.~inv.}}&{=} H(A_1| B_1)_{\rho} \\
\overset{\textnormal{\textnormal{\Cshref{prop_distance}}}}&{\leq} H(A|B)_{\sigma} + 2 \eps_n \log d_A + (1+\eps_n) h\Big(\frac{\eps_n}{1+\eps_n}\Big) \\
&=H(A|B)_{\sigma} + \frac{o(n)}{n}\, , \label{eq_part2_entropy}
\end{align}
where the penultimate step uses the continuity of entropy~\cite{win16}.
Combining~\cref{eq_part1_entropy,eq_part2_entropy} completes the proof.
%%%%%%%%%%%%%%%%%%%%%%%%%%%%%%%%%%%%%%%%%%%%%%%%%%%%%%%%%%%%%%%%%%%%
%%%%%%%%%%%%%%%%%%%%%%%%%%%%%%%%%%%%%%%%%%%%%%%%%%%%%%%%%%%%%%%%%%%%

%%%%%%%%%%%%%%%%%%%%%%%%%%%%%%%%%%%%%%%%%%%%%%%%%%%%%%%%%%%%%%%%%%%%
%%%%%%%%%%%%%%%%%%%%%%%%%%%%%%%%%%%%%%%%%%%%%%%%%%%%%%%%%%%%%%%%%%%%
\section{Robustness of information measures for almost-iid states} \label{sec_entag_meas}
In this work, we justified the importance of almost-iid states. This prompts the question if almost-iid states are as effective as perfect iid states for information-processing tasks. To answer this, it is crucial to understand if certain functionals (that characterize specific information-processing tasks) behave equally or differently for almost-iid and perfect iid states.

In~\cref{sec_cond_entropy} we have seen that the conditional entropy is robust for almost-iid states in the sense that it asymptotically coincides with the entropy of iid states. 
This implies that also the mutual information is robust for almost-iid states. To make this precise, recall that for a bipartite density matrix $\sigma_{AB}$ the mutual information is defined as
\begin{align} \label{eq_def_MI}
I(A:B)_{\sigma}:=H(A)_{\sigma} - H(A|B)_{\sigma} \, .
\end{align}
The mutual information is a popular correlation measure in the sense that it satisfies (i) $I(A:B)_{\sigma} \geq 0$, (ii) $I(A:B)_{\sigma} = 0$ iff $\sigma_{AB}=\sigma_A \otimes \sigma_B$ and (iii) $I(A:BC)_{\sigma} \geq I(A:B)$.\footnote{Properties (i) and (ii) follow by noting that $I(A:B)_{\rho} = D(\rho_{AB} \| \rho_A \otimes \rho_B)$. The third property follows from strong subadditivity together with the chain rule as $I(A:BC)_{\sigma}=I(A:B)_{\sigma} + I(A:C|B)_{\sigma} \geq I(A:B)_{\sigma}$.} 
Let $\sigma_{AB} \in \St(\cH_{AB})$ and $\rho_{A_1^n B_1^n} \in \St^n(\cH_{AB}, \sigma_{AB}^{\otimes n-r})$ for $r=o(n)$. Then 
\begin{align}
 \frac{1}{n} I(A_1^n : B_1^n)_{\rho} 
\overset{\textnormal{\Cshref{eq_def_MI}}}&{=}   \frac{1}{n} H(A_1^n)_{\rho} - \frac{1}{n} H(A_1^n | B_1^n)_{\rho} \\
\overset{\textnormal{\Cshref{thm_entropy_of_almost_iid}}}&{=} H(A)_{\sigma} - H(A|B)_{\sigma} + \frac{o(n)}{n} \\
\overset{\textnormal{\Cshref{eq_def_MI}}}&{=} I(A:B)_{\sigma} + \frac{o(n)}{n} \, .
\end{align}

It is natural to ask if popular entanglement measures are also robust for almost-iid states. In abstract terms, let $E(\cdot)$ be an arbitrary entanglement measure. Let $\sigma_{AB} \in \St(\cH_{AB})$ and $\rho_{A_1^n B_1^n} \in \St^n(\cH_{AB}, \sigma_{AB}^{\otimes n-r})$ for $r=o(n)$. Is it true that 
\begin{align}
\frac{1}{n} E(A_1^n :B_1^n )_{\rho} \overset{?}{=} \frac{1}{n} E(A_1^n :B_1^n )_{\sigma^{\otimes n}} + \frac{o(n)}{n}
\end{align}
In the following, we discuss the robustness of (a) squashed entanglement, (b) entanglement distillation, (c) entanglement cost, and (d) relative entropy of entanglement. 

\subsection{Robustness of squashed entanglement}
Above we have seen that the mutual information is robust under almost-iid states. The same argument can be extended to see that the conditional mutual information also coincides for almost-iid and iid states. 
The \emph{squashed entanglement}~\cite{CW04} is an entanglement measure that is based on the conditional mutual information. Given a biparitite density matrix $\rho_{AB} \in \St(\cH_{AB})$, the squashed entanglement is defined as 
\begin{align}
E_{sq}(A:B)_{\rho}:=\frac{1}{2} \inf_{\rho_{ABE} \in \St(\cH_{ABE})} \big \{I(A:B|E)_{\rho} : \tr_{E}[\rho_{ABE}]=\rho_{AB} \big \} \, ,
\end{align}
where there is no bound on the dimension of $E$. It features many desirable properties such as being additive on tensor products and superadditive in general. We next show that the squashed entanglement for almost-iid and iid states coincide.
\begin{corollary} \label{cor_Squashed_is_robust}
Let $\sigma_{AB} \in \St(\cH_{AB})$ and $\rho_{A_1^n B_1^n} \in \St^n(\cH_{AB}, \sigma_{AB}^{\otimes n-r})$ for $r=o(n)$. Then
\begin{align} \label{eq_question_Esq}
\frac{1}{n} E_{sq}(A_1^n :B_1^n )_{\rho} = E_{sq}(A :B )_{\sigma} + \frac{o(n)}{n} \, .
\end{align}
\end{corollary}
\begin{proof}
Let $d_A := \dim(\cH_A)$ and $d_B := \dim(\cH_B)$.
We can employ the permutation invariance of $\rho$ and the superadditivity of the squashed entanglement~\cite[Proposition~4]{CW04} to write for $\eps_n:=4 \sqrt{\frac{r}{n}}$
\begin{align}
\frac{1}{n} E_{sq}(A_1^n :B_1^n )_{\rho}
\overset{\textnormal{superadditivity}}&{\geq} \frac{1}{n}  \sum_{i=1}^n E_{sq}(A_i:B_i)_{\rho} \\
\overset{\textnormal{perm.~inv.}}&{=}E_{sq}(A:B)_{\rho} \\
\overset{\textnormal{continuity \& \Cshref{prop_distance}}}&{\geq}  E_{sq}(A:B)_{\sigma} - 12 \eps_n \log(d_A d_B) - 6  h(\eps_n) \\
&=  E_{sq}(A:B)_{\sigma} + \frac{o(n)}{n}\, ,
\end{align}
where the continuity of squashed entanglement follows from the continuity of the conditional entropy~\cite{AF03} as explained in~\cite[Section~IV]{CW04}.

It thus remains to prove the other direction. 
For any $\xi >0$ there exists an extension $\sigma_{ABE}$ or $\sigma_{AB}$ with $d_E := \dim(\cH_E) <\infty$ such that 
\begin{align} \label{eq_optimizing_sequence}
\Big|E_{sq}(A:B)_{\sigma} - \frac{1}{2}I(A:B|E)_{\sigma}\Big|\leq \xi \, .
\end{align} 
To see this, note that by the definition of the squashed entanglement there exists an extension $\sigma'_{ABE}$ of $\sigma_{AB}$ (with possibly unbounded $E$-system) such that
\begin{align} \label{eq_sq_def_cont}
\Big|E_{sq}(A:B)_{\sigma} - \frac{1}{2}I(A:B|E)_{\sigma'}\Big|\leq \frac{\xi}{2} \, \, . 
\end{align}
Choose a finite-dimensional projector $\Pi_E$ on the $E$-system such that $\tr[\Pi_E \sigma'_{ABE}]=1-\eps$ for some $\eps>0$. Let
\begin{align}
\sigma_{ABE} := \Pi_E \sigma'_{ABE} \Pi_E^\dagger + \big(1-\tr[\Pi_E \sigma'_{ABE}] \big) \proj{e}_{ABE}\, ,
\end{align}
where $\ket{e}$ is a state orthogonal to the support of $\Pi_E$. 
By the continuity of the conditional entropy~\cite{AF03} we can choose $\eps>0$ such that 
\begin{align} \label{eq_cont_CMI_Esq}
|I(A:B|E)_{\sigma'} - I(A:B|E)_{\sigma}| \leq \xi \, ,
\end{align}
where we used that the continuity of the conditional entropy does not depend on the dimension of the conditioning system.
The triangle inequality implies
\begin{align}
\Big|E_{sq}(A:B)_{\sigma} - \frac{1}{2}I(A:B|E)_{\sigma}\Big|
&\leq \Big|E_{sq}(A:B)_{\sigma} - \frac{1}{2}I(A:B|E)_{\sigma'}\Big| + \Big| \frac{1}{2}I(A:B|E)_{\sigma'} - \frac{1}{2}I(A:B|E)_{\sigma}\Big| \nonumber \\
\overset{\textnormal{\Cshref{eq_sq_def_cont,eq_cont_CMI_Esq}}}&{\leq} \xi \, ,
\end{align}
which thus justifies~\cref{eq_optimizing_sequence}.

Due to~\cref{lem_amost_iid_extension_GM}, there exists an extension $\rho_{A_1^n B_1^n E_1^n}$ of $\rho_{A_1^n B_1^n}$ which is an $\binom{n}{r}$-almost-iid state in $\sigma_{ABE}$. Hence,
\begin{align}
\frac{1}{n} E_{sq}(A_1^n :B_1^n )_{\rho}
&\leq \frac{1}{2n} I(A_1^n:B_1^n|E_1^n)_{\rho} \\
&=\frac{1}{2n} H(A_1^n| E_1^n)_{\rho} - \frac{1}{2n} H(A_1^n|B_1^n E_1^n)_{\rho}  \\
\overset{\textnormal{\Cshref{thm_entropy_of_almost_iid}}}&{=}\frac{1}{2}H(A|E)_{\sigma} - \frac{1}{2}H(A|BE)_{\sigma} + \frac{o(n)}{n} \\
&=\frac{1}{2}I(A:B|E)_{\sigma} + \frac{o(n)}{n} \\
\overset{\textnormal{\Cshref{eq_optimizing_sequence}}}&{\leq}  E_{sq}(A:B)_{\sigma} + \xi  + \frac{o(n)}{n} \, .
\end{align}
Since this holds for any $\xi>0$ we can consider $\xi \to 0$, which concludes the proof.
\end{proof}

\subsection{Robustness of entanglement distillation and entanglement cost}
Let $\ket{\Phi}_{AB}$ denote an entangled Bell state.
Given a bipartite density matrix $\rho_{AB} \in \St(\cH_{AB})$, recall the definitions of \emph{entanglement distillation}~\cite{bennett96,BBPSSW96,bennett_mixed-state_1996} 
\begin{align} \label{eq_def-entanglement-distillation}
E_D(A:B)_{\rho}:= \lim_{\eps \to 0} \lim_{n \to \infty} \sup \Big \{ \frac{m}{n}:  \inf_{\cP_n \in \LOCC} \frac{1}{2} \norm{ \cP_n(\rho^{\otimes n}) - \proj{\Phi}^{\otimes m} }_1 \leq \eps \Big\} 
\end{align}
and \emph{entanglement cost}~\cite{HHT01}
\begin{align}  \label{eq_def-entanglement-cost}
E_C(A:B)_{\rho}:= \lim_{\eps \to 0} \lim_{n \to \infty} \inf \Big \{ \frac{m}{n} : \inf_{\cP_n \in \LOCC} \frac{1}{2} \norm{ \cP_n(\proj{\Phi}^{\otimes m}) - \rho^{\otimes n} }_1 \leq \eps \Big\}  \, .
\end{align}

\begin{question}
Let $\sigma_{AB} \in \St(\cH_{AB})$ and $\rho_{A_1^n B_1^n} \in \St^n(\cH_{AB}, \sigma_{AB}^{\otimes n-r})$ for $r=o(n)$. Is it true that  
\begin{align} \label{eq_question_ED}
\frac{1}{n} E_D(A_1^n :B_1^n )_{\rho} \overset{?}{=} E_D(A :B )_{\sigma} + \frac{o(n)}{n}
\end{align}
\end{question}
\noindent 
As discussed in~\cite{MR_future}, there are strong indications that one direction of~\cref{eq_question_ED} holds, namely that
\begin{align} \label{eq_GR_future}
\frac{1}{n} E_D(A_1^n :B_1^n )_{\rho} \geq E_D(A :B )_{\sigma} + \frac{o(n)}{n}\, .
\end{align}
Whether the other direction holds also remains an open question.

The equivalent question for the entanglement cost asks:
\begin{question}
Let $\sigma_{AB} \in \St(\cH_{AB})$ and $\rho_{A_1^n B_1^n} \in \St^n(\cH_{AB}, \sigma_{AB}^{\otimes n-r})$ for $r=o(n)$. Is it true that 
\begin{align} \label{eq_open_question_EC}
\frac{1}{n} E_C(A_1^n :B_1^n )_{\rho} \overset{?}{=} E_C(A :B )_{\sigma} + \frac{o(n)}{n}
\end{align}
\end{question}
\noindent
However, none of the two directions of~\cref{eq_open_question_EC} are known to hold.

At this point, we emphasize that, unlike squashed entanglement, already the definitions of both entanglement cost and entanglement distillation rely on a tensor power (iid) structure. This may suggest that, rather than asking about the robustness of~\cref{eq_def-entanglement-distillation,eq_def-entanglement-cost}, one should incorporate robustness directly into the definition itself. One may therefore wonder how these notions would change if an almost-iid structure were built into the definition from the outset. In~\cite{MR_future}, the authors investigate this question by introducing new asymptotic state transformation rates that avoid the standard iid assumption. We refer the interested reader to that paper for further details.

\subsection{Robustness of relative entropy of entanglement }
A popular measure to quantify the amount of entanglement is the \emph{relative entropy of entanglement}~\cite{VPRK97} defined as
\begin{align}  \label{eq_def_RE}
E_R(A:B)_{\rho}:= \min_{\sigma_{AB} \in \mathrm{SEP}(A:B)} D(\rho_{AB} \| \sigma_{AB}) \, ,
\end{align}
where $\mathrm{SEP}(A:B):=\conv\{ \proj{\phi}_A \otimes \proj{\varphi}_B : \ket{\phi}_A \in \cH_A, \ket{\varphi}_B \in \cH_B, \braket{\phi|\phi}=\braket{\varphi|\varphi}=1\}$.
It is known~\cite{VW01} that the relative entropy of entanglement is not additive under the tensor product, which justifies the definition of a regularized version $E^\infty_R(A:B)_{\rho}:=\lim_{k \to \infty} \frac{1}{k} E_R(A_1^k:B_1^k)_{\rho^{\otimes k}}$. The limit in the regularization exists due to Fekete's subadditivity lemma.
\begin{question} \label{question_ER}
Let $\sigma_{AB} \in \St(\cH_{AB})$ and $\rho_{A_1^n B_1^n} \in \St^n(\cH_{AB}, \sigma_{AB}^{\otimes n-r})$ for $r=o(n)$. Is it true that 
\begin{align} \label{eq_ER_question}
\frac{1}{n} E_R(A_1^n :B_1^n )_{\rho} \overset{?}{=} \frac{1}{n} E_R(A_1^n :B_1^n )_{\sigma^{\otimes n}} + \frac{o(n)}{n}
\end{align}
\end{question}

If~\cref{eq_ER_question} were true, this would save the original proof of the generalized quantum Stein's lemma~\cite{haya_stein_25,ludo25} by Brand\~ao and Plenio~\cite{brandao_Stein_10} (see also~\cite{Berta2023gapinproofof,stein_nature_24}).
One direction of~\cref{eq_ER_question} follows from the results developed in this paper. To see this, choose $s_n$ a monotonically increasing sequence of integers such that $s_n r = o(n)$ and $\lim_{n \to \infty} s_n = \infty$. Let $k_n = \lceil n /s_n \rceil$.
Note that $n \leq s_n k_n \leq n+s_n$.
%\leq n + s_n$ and therefore $n \leq s_n(k_n +1)$. 
Hence, using the monotonicity of the entanglement of formation under partial trace,
\begin{align}
\frac{1}{n} E_R(A_1^n: B_1^n)_{\rho}
&\leq \frac{1}{n} E_R\big(A_1^{s_n k_n}: B_1^{s_n k_n}\big)_{\rho} \\
&\leq \frac{1}{s_n k_n } E_R\big(A_1^{s_n k_n}: B_1^{s_n k_n}\big)_{\rho} + \frac{s_n}{n}\log d_{A} d_B \, , \label{eq_step_ER_0}
\end{align}
where the final step uses $E_R(A^m:B^m)_{\rho} \leq m \log(d_A d_B)$ and $s_n = o(n)$. In the following steps, we omit the subscripts $n$ for better readability.
Let $\omega_s \in \arg \min_{\tau_{A_1^s B_1^s} \in \mathrm{SEP}} D(\rho_{A_1^s B_1^s} \| \tau_{A_1^s B_1^s})$. Then for $\rho_s = \tr_{n-s}[\rho_n]$ we have
\begin{align}
\frac{1}{n} E_R(A_1^n: B_1^n)_{\rho}
\overset{\textnormal{\Cshref{eq_step_ER_0}}}&{\leq} \frac{1}{sk} E_R\big(A_1^{sk}: B_1^{sk}\big)_{\rho} + \frac{o(n)}{n} \\
 &= \frac{1}{sk} \min_{\tau_{A_1^{sk} B_1^{sk}} \in \mathrm{SEP}} D(\rho_{A_1^{sk} B_1^{sk} } \| \tau_{A_1^{sk} B_1^{sk}}) + \frac{o(n)}{n} \\
 &\leq \frac{1}{sk} D\big(\rho_{sk} \| (\omega_s)^{\otimes k} \big) + \frac{o(n)}{n} \\
 &= -\frac{1}{sk} H(\rho_{sk} ) - \frac{1}{s} \tr[\rho_s \log \omega_s] + \frac{o(n)}{n} \\
  \overset{\textnormal{\Cshref{thm_entropy_of_almost_iid}}}&{=} -\frac{1}{s} H(\sigma^{\otimes s}) - \frac{1}{s} \tr[\rho_s \log \omega_s]  + \frac{o(n)}{n} \\
 \overset{\textnormal{continuity \& \Cshref{prop_distance}}}&{\leq}\frac{1}{s} H(\rho_s) - \frac{1}{s} \tr[\rho_s \log \omega_s] + \frac{o(n)}{n}  \\
 &=\frac{1}{s} E_R(A_1^s: B_1^s)_{\rho} +  \frac{o(n)}{n} \\
 \overset{\textnormal{continuity \& \Cshref{prop_distance}}}&{\leq}\frac{1}{s} E_R(A_1^s: B_1^s)_{\sigma^{\otimes s}} +  \frac{o(n)}{n} \\
 &=\frac{1}{n} E_R(A_1^n: B_1^n)_{\sigma^{\otimes n}} +  \frac{o(n)}{n} \, ,
\end{align}
 where the continuity of the von Neumann entropy and the relative entropy of entanglement can be found in~\cite{audenaert07,petz_Entropybook,win16}. In the last step, we use that $\lim_{n \to \infty} s_n = \infty$ and that the limit in the RHS of~\cref{eq_ER_question} exists due to Fekete's subadditivity lemma.
 
 The other direction of~\cref{eq_ER_question} appears more complicated and remains an open question.

\paragraph{Acknowledgements} We thank Fernando Brand\~ao for his talk on almost-iid states at the SwissMAP Research Station (SRS) conference in Les Diablerets 2024, which motivated us to write this paper. We further thank Frédéric Dupuis and Ludovico Lami for insightful discussions on this topic at the same conference.
GM and RR acknowledge support from the NCCR SwissMAP, the ETH Zurich Quantum Center, the SNSF project No.~20QU-1 225171, and the CHIST-ERA project MoDIC.

%%%%%%%%%%%%%%%%%%%%%%%%%%%%%%%%%%%%%%%%%%%%%%%%%%%%%%%%%%%%%%%%%%%%
%%%%%%%%%%%%%%%%%%%%%%%%%%%%%%%%%%%%%%%%%%%%%%%%%%%%%%%%%%%%%%%%%%%%
\appendix
\section*{Appendix}
%%%%%%%%%%%%%%%%%%%%%%%%%%%%%%%%%%%%%%%%%%%%%%%%%%%%%%%%%%%%%%%%%%%%
\section{Justification of~\cref{rmk_fernando_def}} \label{app_justification_no_fernando}
% Recall the definition from~\cite{brandao_Stein_10} to call a (mixed) density matrix $\rho_{A_1^n}$ almost-iid along $\sigma_A$ if there exist purifications $\ket{\psi}_{A_1^n E_1^n}$ and $\ket{\theta}_{AB}$ of $\rho_{A_1^n}$ and $\sigma_A$, respectively, such that $\ket{\psi}_{A_1^n E_1^n} \in \mathrm{Sym}^n(\cH_{AE}) \cap \mathrm{span}(\cV(\cH_{AB}^{\otimes n},\ket{\theta}_{AB}^{\otimes n -r}))$. 
To justify the assertion of~\cref{rmk_fernando_def}, we need to show that for all purifications $\ket{\psi}_{A_1^2 E_1^2}$ of $\rho_{A_1^2}$ and for all purifications $\ket{\theta}_{AE}$ of $\ket{0}_A$ it follows that $\ket{\psi}_{A_1^2 E_1^2} \not \in \mathrm{Sym}^2(\cH_{AE}) \cap \mathrm{span}\,\cV(\cH_{AE}^{\otimes 2},\ket{\theta}_{AE})$

To see this note that, because $\ket{0}_A$ is pure, $\ket{\theta}_{AE} = \ket{0}_A \otimes \ket{\vartheta}_E$. 
% \begin{align}
% \mathrm{span}\big(\cV(\cH_{AE}^{\otimes 2},\ket{\theta}_{AE})\big)
% &=\mathrm{span} \{ \ket{\theta}_{AE}\ket{0}_A \ket{0}_E,\ket{\theta}_{AE}\ket{0}_A \ket{1}_E, \ldots, \ket{\theta}_{AE}\ket{0}_A \ket{d_E -1}_E,  \nonumber \\
% &\hspace{13.5mm} \ket{\theta}_{AE}\ket{1}_A \ket{0}_E,\ket{\theta}_{AE}\ket{1}_A \ket{1}_E, \ldots, \ket{\theta}_{AE}\ket{1}_A \ket{d_E -1}_E,  \nonumber \\
% &\hspace{13.5mm} \ket{0}_A \ket{0}_E\ket{\theta}_{AE},\ket{0}_A \ket{1}_E\ket{\theta}_{AE}, \ldots, \ket{0}_A \ket{d_E-1}_E\ket{\theta}_{AE},  \nonumber \\
% &\hspace{13.5mm} \ket{1}_A \ket{0}_E\ket{\theta}_{AE},\ket{1}_A \ket{1}_E\ket{\theta}_{AE}, \ldots, \ket{1}_A \ket{d_E-1}_E\ket{\theta}_{AE} \} \, .
% \end{align}
Similarly, because $\rho_{A_1^2} = \proj{\Psi^{-}}$ is pure, any purification on $E$ must be of the form $\ket{\Psi^{-}}_{A_1 A_2} \otimes \ket{\varphi}_{E_1 E_2}=:\ket{\psi}_{A_1^2 E_1^2}$. 
To ensure that $\rho_{A_1^2}$ is an almost-iid state according to the definition from~\cite{brandao_Stein_10}, we need $\ket{\psi}_{A_1^2 E_1^2} \in \mathrm{Sym}^2(\cH_{AE}) \cap \mathrm{span}\,\cV(\cH_{AE}^{\otimes 2},\ket{\theta}_{AE})$. Thus, the swap operation $\pi$ must satisfy
\begin{align}
\pi \ket{\psi}_{A_1^2 E_1^2} 
= (\pi \ket{\Psi^{-}}_{A_1 A_2} ) \otimes (\pi \ket{\varphi}_{E_1 E_2}) 
= - \ket{\Psi^{-}}_{A_1 A_2} \otimes (\pi \ket{\varphi}_{E_1 E_2})
\overset{!}{=} \ket{\Psi^{-}}_{A_1 A_2} \otimes  \ket{\varphi}_{E_1 E_2} \, .
\end{align}
This yields $\pi \ket{\varphi}_{E_1 E_2} = - \ket{\varphi}_{E_1 E_2}$, i.e., $\ket{\varphi}$ is anti-symmetric. Expressing this in an orthonormal basis $\{\ket{i}_E \}_i$ of $E$ with $\ket{0}_E=\ket{\vartheta}_E$ gives $\ket{\varphi}_{E_1 E_2}= \sum_{i,j=0}^{d_E-1} \alpha_{i,j} \ket{i}_{E_1} \ket{j}_{E_2}$ with $\alpha_{i,j}=-\alpha_{j,i}$ for all $i,j \in \{ 0,\ldots,d_E-1\}$. Hence,
\begin{align}
\ket{\psi}_{A_1^2 E_1^2} &= \frac{1}{\sqrt{2}} \sum_{i<j} \alpha_{i,j} \big( \ket{0}_{A_1} \ket{i}_{E_1} \ket{1}_{A_2} \ket{j}_{E_2} - \ket{0}_{A_1} \ket{j}_{E_1} \ket{1}_{A_2} \ket{i}_{E_2} \nonumber \\
&\hspace{25mm} - \ket{1}_{A_1} \ket{i}_{E_1} \ket{0}_{A_2} \ket{j}_{E_2} + \ket{1}_{A_1} \ket{j}_{E_1} \ket{0}_{A_2} \ket{i}_{E_2}\big) \, ,
\end{align}
which cannot be inside $\mathrm{span}\,\cV(\cH_{AE}^{\otimes 2},\ket{\theta}_{AE})$. Thus, the state is not an almost-iid state according to the definition from~\cite{brandao_Stein_10}.

%%%%%%%%%%%%%%%%%%%%%%%%%%%%%%%%%%%%%%%%%%%%%%%%%%%%%%%%%%%%%%%%%%%%
\section{Properties of almost-iid states} \label{app_almost_iid_mixed}
\begin{lemma} \label{eq_def_almost_iid_mixed_states}
Let $\rho_{A_1^n} \in \St^n(\cH_A,\sigma_A^{\otimes n -r})$ with purification $\ket{\theta}_{AE}$ of $\sigma_A$ according to~\cref{def_almost_product_state_mixed}. 
Then, any other purification $\ket{\tilde \theta}_{AR}$ of $\sigma_A$ would also satisfy~\cref{def_almost_product_state_mixed}.
\end{lemma}
\begin{proof}
Without loss of generality, assume $\dim(E)=\dim(R)$. This can be done since the rank of the reduced state of any purification on the purifying system is bounded by the rank of $\sigma_A$.
Since all purifications are then equal up to unitaries on the purifying system, we have $\ket{\tilde \theta}_{AE} = (\id_A \otimes U_E)\ket{\theta}_{AE}$ for some unitary $U_E$ on $E$. For $\rho_{A_1^n E_1^n}$ being the extension of $\rho_{A_1^n}$ according to~\cref{def_almost_product_state_mixed} we define
\begin{align}
\tilde \rho_{A_1^n E_1^n}:= (\id_{A_1^n} \otimes U^{\otimes n}_E) \rho_{A_1^n E_1^n }(\id_{A_1^n} \otimes (U_E^\dagger)^{\otimes n}) \, .
\end{align}
Clearly $\tilde \rho_{A_1^n E_1^n}$ is an extension of $\rho_{A_1^n}$, i.e.,~$\tilde \rho_{A_1^n E_1^n} \in \St(\cH_{AE}^{\otimes n})$ and $\tr_{E_1^n}[\tilde \rho_{A_1^n E_1^n}] = \rho_{A_1^n}$. Thus, it remains to show that this extension satisfies the two properties in~\cref{def_almost_product_state_mixed}.

First, we note that $\tilde \rho_{A_1^n E_1^n}$ is permutation-invariant. To see this, let $\pi$ be a permutation that swaps $(A_i,E_i) \leftrightarrow (A_j,E_j)$. Then
\begin{align}
\pi \tilde \rho_{A_1^n E_1^n} \pi^\dagger 
&=(\id_{A_1^n} \otimes U^{\otimes n}_E) \pi \rho_{A_1^n E_1^n } \pi^\dagger  (\id_{A_1^n} \otimes (U_E^\dagger)^{\otimes n}) \\
&=(\id_{A_1^n} \otimes U^{\otimes n}_E)  \rho_{A_1^n E_1^n }  (\id_{A_1^n} \otimes (U_E^\dagger)^{\otimes n}) \\
&=\tilde \rho_{A_1^n E_1^n} \, .
\end{align}
Second, we have 
\begin{align}
\tilde \rho_{A_1^n E_1^n} 
&= (\id_{A_1^n} \otimes U^{\otimes n}_E) \rho_{A_1^n E_1^n }(\id_{A_1^n} \otimes (U_E^\dagger)^{\otimes n})\\
&=\sum_{k,\ell \in \cT} \beta_{k,\ell} \underbrace{(\id_{A_1^n} \otimes U^{\otimes n}_E) \ket{\Psi_k}}_{\ket{\tilde \Psi_k}} \underbrace{\bra{\Psi_{\ell}} (\id_{A_1^n} \otimes (U_E^\dagger)^{\otimes n})}_{\bra{\tilde \Psi_{\ell}}} \, .  
\end{align}
Note that $\ket{\tilde \Psi_k},\ket{\tilde \Psi_\ell} \in \cV(\cH_{AE}^{\otimes n},\ket{\tilde \theta}_{AE} )$.
\end{proof}

\begin{lemma} \label{eq_ONB_almost_iid_mixed_states}
There exists an orthonormal basis $\{\ket{\Psi_t}\}_{t \in \cT}$ of $\mathrm{span}\,\cV(\cH_{AE}^{\otimes n},\ket{\theta}^{\otimes n-r})$ with vectors $\ket{\Psi_t} \in \cV(\cH_{AE}^{\otimes n},\ket{\theta}^{\otimes n-r})$ for all $t \in \cT$, and with
\begin{align} \label{eq_lemma_sizeT}
    |\cT| \leq \binom{n}{r}\,d_{AE}^{\,r} \leq 2^{n h(r/n)}\, d_{AE}^{\,r},
\end{align}
where $d_{AE} = \dim(\cH_{AE})$, and $h(x):=-x\log x - (1-x)\log(1-x)$ for $x \in [0,1]$ is the binary entropy function.
\end{lemma}

\begin{proof}
Let $\{\ket{x}\}_{x \in \cX}$ denote an orthonormal basis of $\cH_{AE}$ such that $\ket{\theta}=\ket{\bar{x}}$ for some $x \in \cX$. Any vector $\ket{\Psi}\in \mathrm{span}\,\cV(\cH_{AE}^{\otimes n},\ket{\theta}^{\otimes n-r})$ can be expanded as
\begin{align}
    \ket{\Psi} = \sum_j \alpha_j \ket{\tilde{\Psi}_j}  \, ,
\end{align}
for coefficients $\alpha_j \in \C$ and vectors \smash{$\ket{\tilde{\Psi}_j} \in \cV(\cH_{AE}^{\otimes n},\ket{\theta}^{\otimes n-r})$}. By definition, \smash{$\ket{\tilde{\Psi}_j}= \pi_j(\ket{\theta}^{\otimes n-r} \otimes \ket{\Omega_j^{(r)}})$} for some permutation $\pi_j \in \cS_n$ and vector \smash{$\ket{\Omega_j^{(r)}} \in \cH_{AE}^{\otimes r}$}. Since $\{\ket{x}\}_{x \in \cX}$ is a basis of $\cH_{AE}$, we can write 
\begin{align}
\ket{\Omega_j^{(r)}} = \sum_{\bf{x}} \beta^{(j)}_{\bf{x}} \ket{\bf{x}}
\end{align}
for all $j$, where ${\bf{x}} = (x_1,\dots,x_r)$ is an $r$-tuple of elements from $\cX$, $\ket{{\bf{x}}} = \ket{x_1}\otimes \cdots \otimes \ket{x_r}$,
and $\beta^{(j)}_{\bf{x}} \in \C$. Then
\begin{align} \label{eq_lemma_ONB_expansion}
\ket{\Psi} = \sum_{j,\bf{x}} \alpha_j \beta^{(j)}_{\bf{x}}  \,\pi_j(\ket{\theta}^{\otimes n-r} \otimes \ket{{\bf{x}}}) =: \sum_{t} \gamma_t \ket{\Psi_t} 
\end{align}
with \smash{$t = (j,{\bf x})$, $\gamma_t =\alpha_j \beta^{(j)}_{\bf{x}} $, and $\ket{\Psi_t} = \pi_j(\ket{\theta}^{\otimes n-r} \otimes \ket{{\bf{x}}})$}. Intuitively, $j$ is labeling the positions of the defects, and ${\bf x}$ labels the state at these defected positions. Note that the vectors $\ket{\Psi_t}$ are normalized for all $t$, and, for different values of $t$, they are either pairwise orthogonal, or they are equal (e.g.~if $ \ket{{\bf{x}}} = \ket{\theta}^{\otimes r}$). Thus, restricting to a maximal set of pairwise orthogonal vectors $\{\ket{\Psi_t}\}_{t\in \cT}$,~\cref{eq_lemma_ONB_expansion} shows that this set forms an orthonormal basis of $\mathrm{span}\,\cV(\cH_{AE}^{\otimes n},\ket{\theta}^{\otimes n-r})$.

To determine the size $|\cT|$\footnote{More precisely, one finds $|\cT| = \sum_{k=0}^r \binom{n}{k}(d_{AE}-1)^k$. Since $\binom{n}{k} \leq \binom{n}{k}\binom{n-k}{r-k} = \binom{n}{r}\binom{r}{k}$ for all $k\leq r$, it follows that $|\cT| \leq \binom{n}{r} \sum_{k=0}^r \binom{r}{k}(d_{AE}-1)^k = \binom{n}{r} d_{AE}^r $.}, note that for any vector in $\cV(\cH_{AE}^{\otimes n},\ket{\theta}^{\otimes n-r})$ there are $\binom{n}{r}$ possible combinations for the positions of the defects, and at each such position, the dimension of the Hilbert space is given by $d_{AE}$. Hence, we get the upper bound stated in~\cref{eq_lemma_sizeT}, where the first bound would correspond to the case where the defects are different from $\ket{\theta}$. The second inequality then follows from the relation $\binom{n}{r} \leq 2^{n h(r/n)}$, e.g.~see~\cite[Example 11.1.3]{cover}.
\end{proof}

\begin{lemma} \label{fact_Giulia_proof}
$\rho_{A_1^{n+m}} \in \St^{n+m}(\cH_A, \sigma_A^{\otimes n+m-r})$  implies $\rho_{A_1^n} \in \St^n(\cH_A, \sigma_A^{\otimes n-r})$ for any $m,n \in \N$.
\end{lemma}
\begin{proof}
Since $\rho_{A_1^{n+m}} \in \St^{n+m}(\cH_A, \sigma_A^{\otimes n+m-r})$, there exists an extension $\rho_{A_1^{n+m} E_1^{n+m}}$ that is permutation invariant. 
We first show that $\rho_{A_1^n E_1^n}$ is then also permutation-invariant. To see this, let $\pi \in \cS_n$ denote an arbitrary permutation. Then
\begin{align}
\pi \rho_{A_1^n E_1^n} \pi^\dagger 
&=\pi  \tr_{A_{n+1}^{n+m}E_{n+1}^{n+m}}[\rho_{A_1^{n+m}E_1^{n+m}}] \pi^\dagger \\
&=\tr_{A_{n+1}^{n+m}E_{n+1}^{n+m}}[ (\pi \otimes \id_{m}) \rho_{A_1^{n+m}E_1^{n+m}} (\pi^\dagger \otimes \id_{m}) ] \\
&= \tr_{A_{n+1}^{n+m}E_{n+1}^{n+m}}[\rho_{A_1^{n+m}E_1^{n+m}} ] \\
&=\rho_{A_1^n E_1^n} \, .
\end{align}
Thus, it remains to show that $\rho_{A_1^n E_1^n}$ satisfies Property~\eqref{it_second_def} in~\cref{def_almost_product_state_mixed}.
By definition \smash{$\rho_{A_1^{n+m} E_1^{n+m}}$} is such that $\supp(\rho_{A_1^{n+m} E_1^{n+m}}) \subseteq \mathrm{span} \, \cV(\cH^{\otimes n+m}_{AE},\ket{\theta}_{AE}^{\otimes n+m-r})$ for a purification $\ket{\theta}_{AE}$ of $\sigma_A$. We need to show that $\supp(\rho_{A_1^{n} E_1^{n}}) \subseteq \mathrm{span} \, \cV(\cH^{\otimes n}_{AE},\ket{\theta}_{AE}^{\otimes n-r})$. To see this, let $\{\ket{\psi_t} \}_{t \in \cT}$ be the ``standard" tensor product basis of $\cV(\cH^{\otimes n+m}_{AE},\ket{\theta}_{AE}^{\otimes n+m-r})$, i.e., for a basis $\{\ket{u} \}_u$ of $\cH_{AE}$ we have $\ket{\psi_t} = \pi_t (\ket{u_1}_{A_1 E_1} \otimes \ldots \otimes  \ket{u_r}_{A_r E_r} \otimes \ket{\theta}_{AE}^{\otimes n+m-r})$ for some permutation $\pi_t$. We can observe that
\begin{align} \label{eq_Giulia0}
\rho_{A_1^n E_1^n} = \tr_{A_{n+1}^{n+m}E_{n+1}^{n+m}}[\rho_{A_1^{n+m} E_1^{n+m}}] = \sum_{k,\ell \in \cT} \beta_{k,\ell} \tr_{A_{n+1}^{n+m}E_{n+1}^{n+m}}[ \ket{\psi_k}\bra{\psi_\ell} ]
\end{align}
and
\begin{align} \label{eq_Giulia1}
\tr_{A_{n+1}^{n+m}E_{n+1}^{n+m}}[ \ket{\psi_k}\bra{\psi_\ell} ] = \sum_{s,t \in \bar \cT } c^{(k,\ell)}_{s,t}  \ket{\bar \psi_s}\bra{ \bar \psi_t} \, ,
\end{align}
where $\{\ket{\bar \psi_t} \}_{t \in \bar \cT}$ denotes the standard basis of $\cV(\cH^{\otimes n}_{AE},\ket{\theta}_{AE}^{\otimes n-r})$. Note that these basis elements in~\cref{eq_Giulia1} are independent of $k$ and $\ell$. Hence, we find
\begin{align}
\rho_{A_1^n E_1^n} 
\overset{\textnormal{\Cshref{eq_Giulia0}}}&{=} \sum_{k,\ell \in \cT} \beta_{k,\ell} \tr_{A_{n+1}^{n+m}E_{n+1}^{n+m}}[ \ket{\psi_k}\bra{\psi_\ell} ] \\
\overset{\textnormal{\Cshref{eq_Giulia1}}}&{=} \sum_{k,\ell \in \cT, s,t \in \bar \cT} \beta_{k,\ell} c^{(k,\ell)}_{s,t}   \ket{ \bar \psi_s}\bra{\bar \psi_t}  \\
&= \sum_{s,t \in \bar \cT} \gamma_{s,t}   \ket{ \bar \psi_s}\bra{\bar \psi_t} \, ,
\end{align}
where the final step uses $\gamma_{s,t}:=\sum_{k,\ell \in \cT} \beta_{k,\ell} c^{(k,\ell)}_{s,t}$.
This shows that $\supp(\rho_{A_1^{n} E_1^{n}}) \subseteq \mathrm{span} \, \cV(\cH^{\otimes n}_{AE},\ket{\theta}_{AE}^{\otimes n-r})$ and hence completes the proof.
\end{proof}

\begin{lemma} \label{lem_PERM_almost_iid}
Let $\rho_{A_1^n} \in \bar \St^n(\cH_A, \sigma_A^{\otimes n-r})$. Then for any $s \in \N$ we have $(\rho_{A_1^n})^{\otimes s} \in \bar \St^{ns}(\cH_A, \sigma_A^{\otimes ns-rs})$. 
\end{lemma}
\begin{proof}
Let $\ket{\theta}_{AE}$ denote a purification of $\sigma_A$. Since $\rho_{A_1^n} \in \bar \St^n(\cH_A, \sigma_A^{\otimes n-r})$ there exists an extension $\rho_{A_1^n E_1^n}$ that can be written as
\begin{align} \label{eq_decomp_PI}
\rho_{A_1^n E_1^n} = \sum_{i,j} \beta_{i,j} \ket{\Psi_i} \bra{\Psi_j} \, ,
\end{align}
for $\ket{\Psi_t} \in \cV(\cH_{AE}^{\otimes n},\ket{\theta}^{\otimes n-r})$. Furthermore, by~\cref{eq_decomp_PI} we have
\begin{align}
(\rho_{A_1^n E_1^n})^{\otimes s} = \sum_{i_1,j_1,\ldots,i_s,j_s} \beta_{i_1,j_1} \ldots \beta_{i_s,j_s} \ket{\Psi_{i_1}} \bra{\Psi_{j_1}} \otimes \ldots \otimes \ket{\Psi_{i_s}} \bra{\Psi_{j_s}} \, ,
\end{align}
where $\ket{\Psi_{i_1}} \otimes \ldots \otimes \ket{\Psi_{i_s}} \in \cV(\cH_{AE}^{\otimes n s},\ket{\theta}^{\otimes ns -rs})$. This shows that $\supp(\rho_{A_1^n E_1^n}^{\otimes s}) \subseteq  \cV(\cH_{AE}^{\otimes n s},\ket{\theta}^{\otimes ns -rs})$, which completes the proof. %\qed
\end{proof}

\begin{lemma} \label{lem_amost_iid_extension_GM}
Let $n \in \N$, $r \leq n$, $\sigma_{ABE} \in \St(A \otimes B \otimes E)$, and $\rho_{A_1^n B_1^n} \in \St^n(\cH_{AB},\sigma_{AB}^{\otimes n -r})$. Then, there exists an extension $\rho_{A_1^n B_1^n E_1^n}$ of $\rho_{A_1^n B_1^n}$ that is $\binom{n}{r}$-almost-iid in $\sigma_{ABE}$, i.e.~$\rho_{A_1^n B_1^n E_1^n} \in \St^n(\cH_{ABE},\sigma_{ABE}^{\otimes n -r})$. 
\end{lemma}
\begin{proof} This is essentially a consequence of~\cref{eq_def_almost_iid_mixed_states}. Formally, there exist a purification $\ket{\theta}_{ABG}$ of $\sigma_{AB}$ and an extension $\rho_{A_1^n B_1^nG_1^n}$ of $\rho_{A_1^n B_1^n}$ such that $\supp(\rho_{A_1^nB_1^nG_1^n}) \subseteq  \mathrm{span}\,\cV(\cH_{ABG}^{\otimes n},\ket{\theta}_{ABG}^{\otimes n-r})$, by definition. Now fix a purification $\ket{\tilde{\theta}}_{ABEF}$ of $\sigma_{ABE}$ such that $\dim(EF) \geq \dim(G)$. Since $\ket{\tilde{\theta}}_{ABEF}$ is also a purification of $\sigma_{AB}$, there exists an isometry $V:\,\cH_G \to \cH_{EF}$ such that $V \ket{\theta}_{ABG} =\ket{\tilde{\theta}}_{ABEF}$. Then we define
\begin{align}
\tilde \rho_{A_1^n B_1^n E_1^n F_1^n}:= (\id_{A_1^n B_1^n} \otimes V^{\otimes n}) \rho_{A_1^n B_1^nG_1^n}(\id_{A_1^n  B_1^n} \otimes (V^\dagger)^{\otimes n})\,,
\end{align}
and 
\begin{align}
 \rho_{A_1^n B_1^n E_1^n} := \tr_{F_1^n}[\tilde \rho_{A_1^nB_1^n E_1^n F_1^n}]\,.
\end{align}
Clearly, $\tilde \rho_{A_1^n B_1^n E_1^n F_1^n}$ is an extension of $\rho_{A_1^n  B_1^n}$. Furthermore, this extension satisfies the two properties in~\cref{def_almost_product_state_mixed}, namely, it is permutation invariant and $\supp(\tilde{\rho}_{A_1^nB_1^nE_1^n F_1^n}) \subseteq  \mathrm{span}\,\cV(\cH_{ABEF}^{\otimes n},\ket{\tilde{\theta}}_{ABEF}^{\otimes n-r})$, which can be shown by following the same arguments as in the proof of~\cref{eq_def_almost_iid_mixed_states}. On the other hand, $\tilde \rho_{A_1^n B_1^n E_1^n F_1^n}$ is also an extension of $\rho_{A_1^n B_1^n E_1^n}$, which means that $\rho_{A_1^n B_1^n E_1^n}$ is $\binom{n}{r}$-almost-iid in $\sigma_{ABE}$. Since $\rho_{A_1^n B_1^n E_1^n}$ is an extension of $\rho_{A_1^n  B_1^n}$, the claim follows.
\end{proof}

%%%%%%%%%%%%%%%%%%%%%%%%%%%%%%%%%%%%%%%%%%%%%%%%%%%%%%%%%%%%%%%%%%%%
%%%%%%%%%%%%%%%%%%%%%%%%%%%%%%%%%%%%%%%%%%%%%%%%%%%%%%%%%%%%%%%%%%%%
\section{Proof of~\cref{prop_statistics}} \label{app_proof_statistics}
We generalize the statement for pure almost-iid states from~\cite[Theorem~4.5.2]{renner_phd} to the more general setting of mixed almost-iid states. The proof works similarly to the one from~\cite[Theorem~4.5.2]{renner_phd} with a few modifications. 

Let $\ket{\theta}_{AE}$ be a purification of $\rho$, where $E$ denotes the purifying system of dimension $d = \dim \cH$. Since $\rho^{(n)}$ is an almost-iid state, there exists an extension $\rho_{A_1^n E_1^n}$ which can be written as
\begin{align}
\rho_{A_1^n E_1^n} = \sum_{i,j \in \cT} \beta_{i,j} \ket{\Psi_i} \bra{\Psi_j} \, ,
\end{align}
where $\{\ket{\Psi_t}\}_{t \in \cT}$ is an orthonormal basis of $\mathrm{span}\,\cV(\cH_{AE}^{\otimes n},\ket{\theta}^{\otimes n-r})$ with vectors $\ket{\Psi_t} \in \cV(\cH_{AE}^{\otimes n},\ket{\theta}^{\otimes n-r})$.
For any fixed $t \in \cT$ we can assume without loss of generality that $\ket{\Psi_t} = \ket{\theta}^{\otimes n-r} \otimes \ket{\Omega_r}$, where $\ket{\Omega_r}$ represents the defects. Let $\mathbf{x}=(x_1, \ldots x_n)$ be the outcomes of the measurement $\cM^{\otimes n}$ applied to $\proj{\Psi_t}$. Furthermore, let $\mathbf{x'}=(x_1,\ldots,x_{n-r})$ and $\mathbf{x''}=(x_{n-r+1},\ldots,x_{n})$. Clearly $\mathbf{x'}$ is distributed according to the product distribution $(P_X)^{n-r}$. Hence, for any $\delta >0$ we have
\begin{align}
\PP\Big[\norm{\lambda_{\mathbf{x'}} -P_X}_1 > \sqrt{2 (\ln 2) \Big(\delta + \frac{|\cX| \log(n-r +1)}{n-r} \Big) }\Big] 
\overset{\textnormal{\cite[Cor.~B.3.3]{renner_phd}}}{\leq} 2^{-(n-r)\delta} \, .
\end{align}
Using $r \leq \frac{n}{2}$ this can be simplified to
\begin{align} \label{eq_green_star_stat}
\PP\Big[\norm{\lambda_{\mathbf{x'}} -P_X}_1 > \sqrt{2 (\ln 2) \Big(\delta + \frac{2 |\cX| \log(\frac{n}{2}+1)}{n} \Big) }\Big] 
\leq 2^{-\frac{n\delta}{2}} \, .
\end{align}
Using $\lambda_{\mathbf{x}} = \frac{n-r}{n} \lambda_{\mathbf{x'}} + \frac{r}{n} \lambda_{\mathbf{x''}}$ yields
\begin{align} \label{eq_red_star_stat}
\norm{\lambda_{\mathbf{x}} - P_X}_1
\overset{\textnormal{triangle}}{\leq} \frac{n-r}{n} \norm{\lambda_{\mathbf{x'}} - P_X}_1 + \frac{r}{n} \norm{\lambda_{\mathbf{x''}} - P_X}_1
\leq \norm{\lambda_{\mathbf{x'}} - P_X}_1 + \frac{2r}{n} \, .
\end{align}
Hence,
\begin{align}
&\PP\Big[\norm{\lambda_{\mathbf{x}} -P_X}_1 > \sqrt{2 (\ln 2) \Big(\delta + \frac{2 |\cX| \log(\frac{n}{2}+1)}{n} \Big) } + \frac{2r}{n}\Big] \nonumber \\
&\hspace{35mm}\overset{\textnormal{\Cshref{eq_red_star_stat}}}{\leq} \PP\Big[\norm{\lambda_{\mathbf{x'}} -P_X}_1 > \sqrt{2 (\ln 2) \Big(\delta + \frac{2 |\cX| \log(\frac{n}{2}+1)}{n} \Big) }\Big] \\
&\hspace{35mm}\overset{\textnormal{\Cshref{eq_green_star_stat}}}{\leq}2^{-\frac{n\delta}{2}}\, .
\end{align}
This can be rewritten as
\begin{align} \label{eq_blue_star_stat}
\underset{\mathbf{x} \leftarrow \ket{\Psi_t}}{\PP}[ \mathbf{x} \in \cW_{\delta}] \leq 2^{-\frac{n\delta}{2}} \, ,
\end{align}
for
\begin{align}
\cW_{\delta} = \Big\{ \mathbf{x} \in \cX^n : \norm{\lambda_{\mathbf{x}} - P_X}_1 > \sqrt{2 (\ln 2) \Big(\delta + \frac{2 |\cX| \log(\frac{n}{2}+1)}{n} \Big) } + \frac{2r}{n}  \Big \} \, .
\end{align}
The notation $\mathbf{x} \leftarrow \ket{\Psi_t}$ indicates that $\mathbf{x}$ is distributed according to the outcomes of the measurement applied to $\ket{\Psi_t}$.

For $M_{\mathbf{x}}= M_{x_1} \otimes \ldots \otimes M_{x_n}$ we find
\begin{align}
\underset{\mathbf{x} \leftarrow \rho_{A_1^n E_1^n}}{\PP}[ \mathbf{x} \in \cW_{\delta}]
&= \sum_{\mathbf{x} \in \cW_{\delta}} \tr[\rho_{A_1^n E_1^n} M_{\mathbf{x}} ] \\
\overset{\textnormal{\Cshref{eq_RR_operator_ineq}}}&{\leq} \sum_{\mathbf{x} \in \cW_{\delta}} |\cT| \sum_{t \in \cT} \beta_{t,t} \bra{\Psi_t} M_{\mathbf{x}} \ket{\Psi_t} \\
&=|\cT|\sum_{t \in \cT} \beta_{t,t}  \underset{\mathbf{x} \leftarrow \ket{\Psi_t}}{\PP}[ \mathbf{x} \in \cW_{\delta}] \\
\overset{\textnormal{\Cshref{eq_blue_star_stat}}}&{\leq} |\cT|2^{-\frac{n\delta}{2}} \\
\overset{\textnormal{\Cshref{rmk_properties_almost_iid_states}}}&{\leq} 2^{-n(\frac{\delta}{2} - h(\frac{r}{n}))} d^{2r} \, . \label{eq_square_stats}
\end{align}
Choosing $\delta= \frac{2 \log(\frac{1}{\eps})}{n}+ 2h(\frac{r}{n}) + \frac{4r}{n} \log(d)$ yields
\begin{align}
\PP\Big[\norm{\lambda_{\mathbf{x}} -P_X}_1 > \sqrt{4 (\ln 2) \left(\frac{\log(\frac{1}{\eps})}{n} + h\Big(\frac{r}{n}\Big) + \frac{2r}{n} \log(d) + \frac{|\cX| \log(\frac{n}{2}+1)}{n} \right) }\Big]
\overset{\textnormal{\Cshref{eq_square_stats}}}{\leq} \eps \, ,
\end{align}
which completes the proof. \qed
%%%%%%%%%%%%%%%%%%%%%%%%%%%%%%%%%%%%%%%%%%%%%%%%%%%%%%%%%%%%%%%%%%%%
%%%%%%%%%%%%%%%%%%%%%%%%%%%%%%%%%%%%%%%%%%%%%%%%%%%%%%%%%%%%%%%%%%%%
%%%%%%%%%%%%%%%%%%%%%%%%%%%%%%%%%%%%%%%%%%%%%%%%%%%%%%%%%%%%%%%%%%%%
\section{Proof of~\cref{prop_no_classical_exp_dF}} \label{app_no_classical_dF}
Consider a $(n+k)$-bit string with $m$ ones and $n+k-m$ zeros. If we throw away $k$ bits, then the probability of having $j$ ones in the remaining $n$-bit string is
\begin{align} \label{eq_probability_j_zeros}
p_j = \frac{\binom{n}{j} \binom{k}{m-j}}{\binom{n+k}{m}} \, .
\end{align}
Note that $\sum_{j=0}^m p_j =1$ follows from a known identity due to Vandermonde which states that for all $r,s,t \in \N$ we have
\begin{align} \label{eq_Vandermonde}
\binom{r+s}{t} = \sum_{\ell=0}^t \binom{r}{\ell} \binom{s}{t-\ell} \, .
\end{align}
To see~\cref{eq_probability_j_zeros}, let $v \in \{0,1,\ldots,k\}$ denote the number of ones that have been thrown away. Hence, 
\begin{align}
p_j  = \frac{1}{c} \binom{n}{j} \binom{k}{v} =  \frac{1}{c} \binom{n}{j} \binom{k}{m-j} \, .
\end{align}
for some normalization constant
\begin{align}
c 
= \sum_{j=0}^m \binom{n}{j}\binom{k}{m-j}
\overset{\textnormal{\Cshref{eq_Vandermonde}}}{=}\binom{n+k}{m} \, .
\end{align}
\begin{fact} \label{fact_variance}
 For the setting above, we have
 \begin{align}
 \mathrm{Var}_p[X] =  \frac{kmn(n+k-m)}{(n+k-1)(n+k)^2} \, .
 \end{align}
\end{fact}
\begin{proof}
Recall that
\begin{align}
 j \binom{n}{j} = n \binom{n-1}{j-1} \, . \label{eq_property_binom}
\end{align}
With this we can write
\begin{align}
\E_p[X] 
&= \sum_{j=0}^m j p_j \\
\overset{\textnormal{\Cshref{eq_probability_j_zeros}}}&{=} \sum_{j=0}^m j   \frac{\binom{n}{j} \binom{k}{m-j}}{\binom{n+k}{m}} \\
\overset{\textnormal{\Cshref{eq_property_binom}}}&{=} \frac{n}{\binom{n+k}{m}}   \sum_{j=1}^m \binom{n-1}{j-1} \binom{k}{m-j}\\
&=\frac{n}{\binom{n+k}{m}}   \sum_{j=0}^{m-1} \binom{n-1}{j} \binom{k}{m-1-j}\\
\overset{\textnormal{\Cshref{eq_Vandermonde}}}&{=}\frac{n}{\binom{n+k}{m}}  \binom{n+k-1}{m-1} \\
\overset{\textnormal{\Cshref{eq_property_binom}}}&{=}\frac{nm}{n+k} \, . \label{eq_expecation_value}
\end{align}
Similarly, we find
\begin{align}
\E_p[X^2]
&= \sum_{j=0}^m j^2 p_j \\
\overset{\textnormal{\Cshref{eq_probability_j_zeros}}}&{=} \frac{1}{{\binom{n+k}{m}}} \sum_{j=0}^m j^2   \binom{n}{j} \binom{k}{m-j} \\
&= \frac{n}{{\binom{n+k}{m}}} \sum_{j=0}^{m-1} (j+1)   \binom{n-1}{j} \binom{k}{m-1-j}\\
&= \frac{n}{{\binom{n+k}{m}}} \left( \sum_{j=0}^{m-1} j   \binom{n-1}{j} \binom{k}{m-1-j} + \sum_{j=0}^{m-1} \binom{n-1}{j} \binom{k}{m-1-j} \right)\\
\overset{\textnormal{\Cshref{eq_Vandermonde,eq_property_binom}}}&{=} \frac{n}{{\binom{n+k}{m}}} \left( (n-1) \sum_{j=0}^{m-2}   \binom{n-2}{j} \binom{k}{m-2-j} + \binom{n+k-1}{m-1}  \right) \\
\overset{\textnormal{\Cshref{eq_Vandermonde}}}&{=}\frac{n}{{\binom{n+k}{m}}} \left( (n-1) \binom{n+k -2}{m-2} + \binom{n+k-1}{m-1}  \right) \\
&=n(n-1) \frac{\binom{n+k-1}{m-1} \binom{n+k-2}{m-2}}{\binom{n+k}{m} \binom{n+k-1}{m-1}} + n \frac{\binom{n+k-1}{m-1}}{\binom{n+k}{m}}\\
&= n(n-1)  \frac{ m (m-1) }{(n+k)(n+k-1)} + n \frac{ m}{n+k} \, .  \label{eq_expecation_value2}
\end{align}
Combining everything yields
\begin{align}
 \mathrm{Var}_p[X]
 =\E_p[X^2] - (\E_p[X])^2 
 \overset{\textnormal{\Cshref{eq_expecation_value,eq_expecation_value2}}}{=}  \frac{kmn(n+k-m)}{(n+k-1)(n+k)^2} \, .
\end{align}
\end{proof}

Let $Q^{(q)}_{X_1^n} \in \cV(\cX^n,q^{n-r})$ and consider a binary random string $X_1^n  \sim Q^{(q)}_{X_1^n}$. Then without loss of generality assume that the $r$ defects are at the end of the random string, and hence
\begin{align} \label{eq_simple_variance_term}
 \mathrm{Var}\Big[ \sum_{i=1}^n X_i \Big] 
 =   \mathrm{Var} \Big[ \sum_{i=1}^{n-r} X_i \Big] + \underbrace{  \mathrm{Var} \Big[ \sum_{i=n-r+1}^{n} X_i \Big]}_{\geq 0}
 \geq (n-r) \mathrm{Var}_q[X] \, ,
\end{align}
where the first equality uses that the defects are independent of iid parts. 

\begin{example} \label{ex_Variance_1}
Let $\alpha \in (0,1)$, $m=\frac{n+k}{2}$, $k=n^{\alpha}$ and $r=o(n)$. Then~\cref{fact_variance} gives $ \mathrm{Var}_p[X] = \frac{n^{1+\alpha}}{4(n+n^{\alpha}-1)} = \Theta(n^{\alpha})$. Furthermore,~\cref{eq_simple_variance_term} yields $ \mathrm{Var} \Big[ \sum_{i=1}^n X_i \Big] \geq \Theta(n)$.
\end{example}

\begin{fact} \label{fact_concavity_variance}
Let $t \in [0,1]$ and $p,q$ be two probability distributions. Then,
\begin{align}
 \mathrm{Var}_{tp +(1-t)q}[X] \geq t  \mathrm{Var}_{p}[X] + (1-t)  \mathrm{Var}_{q}[X] \, .
\end{align}
\end{fact}
\begin{proof}
By definition of the variance, we have
\begin{align}
 \mathrm{Var}_{tp +(1-t)q}[X] 
 &= \E_{tp +(1-t)q}[X^2] - (\E_{tp +(1-t)q}[X])^2 \\
 &= t \E_p[X^2] + (1-t) \E_{q}[X^2] - (t \E_p[X] + (1-t) \E_{q}[X])^2 \\
 &\geq t \E_p[X^2] + (1-t) \E_{q}[X^2] -( t \E_p[X]^2 + (1-t) \E_{q}[X]^2) \\
 &=t  \mathrm{Var}_{p}[X] + (1-t)  \mathrm{Var}_{q}[X] \, .
\end{align}
\end{proof}

Putting everything together, we obtain for $\alpha \in (0,1)$
\begin{align}
\mathrm{Var}_p[X] \overset{\textnormal{\Cshref{ex_Variance_1}}}{=}  \Theta(n^{\alpha}) \quad \textnormal{and} \quad 
 \mathrm{Var}_{ \int Q_{X}^{(q)}  \nu(\di q)}[X] \overset{\textnormal{\Cshref{fact_concavity_variance}}}{\geq}  \int  \nu(\di q)   \mathrm{Var}_{ Q_{X}^{(q)}}[X] \overset{\textnormal{\Cshref{ex_Variance_1}}}{=}\Theta(n) \, .
\end{align}
This proves the assertion of~\cref{prop_no_classical_exp_dF}. Note that for better readability, the above proof has been done for the specific choice $k=n^{\alpha} = o(n)$ for $\alpha \in (0,1)$, but the same argument remains valid for an arbitrary $k=o(n)$. \qed

%%%%%%%%%%%%%%%%%%%%%%%%%%%%%%%%%%%%%%%%%%%%%%%%%%%%%%%%%%%%%%%%%%%%
%%%%%%%%%%%%%%%%%%%%%%%%%%%%%%%%%%%%%%%%%%%%%%%%%%%%%%%%%%%%%%%%%%%%
%%%%%%%%%%%%%%%%%%%%%%%%%%%%%%%%%%%%%%%%%%%%%%%%%%%%%%%%%%%%%%%%%%%%
\section{Alternative proof of~\cref{thm_entropy_of_almost_iid}}   \label{app_alternative_proof}
In this section, we present an alternative proof for~\cref{thm_entropy_of_almost_iid} which uses an entirely different proof technique which may be of independent interest. However, we note that the scaling of the defects $r$ is slightly worse than in~\cref{thm_entropy_of_almost_iid}.
Note that it suffices to prove the following result.
\begin{theorem} \label{thm_entropy_of_almost_iid_alternative}
Let $\sigma_{A} \in \St(\cH_{A})$ and $\rho_{A_1^n} \in \St^n(\cH_{A}, \sigma_{A}^{\otimes n-r})$ for $r=o(\sqrt{n})$. Then 
\begin{align} \label{eq_almostProduct_entropy_alt}
 \frac{1}{n} H(A_1^n)_{\rho} = H(A)_{\sigma} + \frac{o(n)}{n} \, .
\end{align}
\end{theorem}

\Cref{thm_entropy_of_almost_iid_alternative} implies that the conditional entropy of almost-iid states coincides asymptotically with the conditional entropy of iid states.
\begin{corollary} \label{lem_cond_entropy_almost_iid}
Let $\sigma_{AB} \in \St(\cH_{AB})$ and $\rho_{A_1^n B_1^n} \in \St^n(\cH_{AB}, \sigma_{AB}^{\otimes n-r})$ for $r=o(\sqrt{n})$. Then 
\begin{align} \label{eq_almostProduct_entropy_cor_alt}
 \frac{1}{n} H(A_1^n|B_1^n)_{\rho} = H(A|B)_{\sigma} + \frac{o(n)}{n} \, .
\end{align}
\end{corollary}
\begin{proof}
By definition of the conditional entropy, we have
\begin{align*}
\frac{1}{n}H(A_1^n|B_1^n)_{\rho}
= \frac{1}{n}H(A_1^n B_1^n)_{\rho} - \frac{1}{n} H(B_1^n)_{\rho} 
\overset{\textnormal{\Cshref{thm_entropy_of_almost_iid_alternative}}}{=} H(AB)_{\sigma}-H(B)_{\sigma} + \frac{o(n)}{n} 
= H(A|B)_{\sigma} + \frac{o(n)}{n} \, .
\end{align*}
\end{proof}

\subsection{Proof of~\cref{thm_entropy_of_almost_iid_alternative}}
One direction of~\cref{eq_almostProduct_entropy_alt} is simple. To see this, let $\eps_n = 2\sqrt{\frac{rs}{n}}$ for $s=\lfloor \sqrt{n} \rfloor$ and consider
\begin{align}
\frac{1}{n} H(A_1^n)_{\rho}
\overset{\textnormal{subadditivity}}{\leq} H(A_1)_{\rho}
\overset{\textnormal{\textnormal{\Cshref{prop_distance}}}}{\leq} H(A)_{\sigma} + \eps_n \log d + h(\eps_n) 
=H(A)_{\sigma} + \frac{o(n)}{n}\, ,
\end{align}
where the penultimate step uses the continuity of entropy~\cite{audenaert07,petz_Entropybook,win16}.

The other direction is more complicated.
Recall that strong subadditivity of quantum entropy (SSA)~\cite{LieRus73_1,LieRus73} ensures $I(A:C|B)_{\rho} \geq 0$. Furthermore, the conditional mutual information satisfies a chain rule $I(A:BC)_{\rho} = I(A:B)_{\rho} + I(A:C|B)_{\rho}$.
Let $\ket{\theta}_{AE}$ be a purification of $\sigma_A$ and let $\rho_{A_1^n E_1^n}$ be an extension of $\rho_{A_1^n}$ that satisfies the two conditions of~\cref{def_almost_product_state_mixed}. Since $\rho_{A_1^n E_1^n}$ is permutation-invariant, for entropy and mutual information terms the indices of the considered subsystems can be changed.
Hence, we find for any $k \leq \lfloor\frac{n}{2}\rfloor=:n_0$, $k'\leq n_0 +1$ and for any $\ell=k,\ldots,n_0$, $\ell'=k',\ldots,n_0 +1$
\begin{align}
I(A_1^{n_0} E_1^{n_0} : A_{n_0+1}^n E_{n_0+1}^n)_{\rho}
\overset{\mathrm{SSA}}&{\geq}I(A_1^{n_0}: A_{n_0+1}^{n})_{\rho} \\
\overset{\mathrm{chain\,rule}}&{=}I(A_1^{\ell}:A_{n_0+1}^{n}|A_{\ell+1}^{n_0})_{\rho} + I(A_{\ell+1}^{n_0}:A_{n_0+1}^{n})_{\rho} \\
\overset{\mathrm{SSA}}&{\geq} I(A_1^k:A_{n_0+1}^{n}|A_{\ell+1}^{n_0})_{\rho}\\
\overset{\mathrm{chain\,rule}}&{=}I(A_1^k : A_{n_0+1}^{n_0+{\ell'}}|A_{\ell+1}^{n_0} A_{n_0+\ell'+1}^{n})_{\rho} + I(A_1^k:A_{n_0+\ell'+1}^n|A_{\ell+1}^{n_0})_{\rho}\\
\overset{\mathrm{SSA}}&{\geq} I(A_1^k : A_{n_0+1}^{n_0+k'}|A_{\ell+1}^{n_0} A_{n_0+\ell'+1}^{n})_{\rho}   \\
\overset{\mathrm{perm.~inv.}}&{=} I(A_1^k : A_{k+1}^{k+k'}|A_{k+k'+1}^m)_{\rho} \, , \label{eq_step_mid_mid}
\end{align}
where in the last step, we permute all the remaining systems appearing in the conditioning (there are $n-\ell-\ell'$ many) into neighboring systems labeled by indices from $(k+k'+1)$ to $m := (k+k' + n-\ell-\ell')$.

Thus, for any $m \geq k + k'$ we find
\begin{align}
I(A_1^k : A_{k+1}^{k+k'}|A_{k+k'+1}^{m})_{\rho}
\overset{\textnormal{\Cshref{eq_step_mid_mid}}}&{\leq}I(A_1^{n_0} B_1^{n_0} : A_{n_0+1}^n E_{n_0+1}^n)_{\rho} \\
&= H(A_1^{n_0} E_1^{n_0})_{\rho} + H(A_{n_0+1}^n E_{n_0+1}^n)_{\rho}- H(A_1^n E_1^n)_{\rho}\\
&\leq H(A_1^{n_0} E_1^{n_0})_{\rho} + H(A_{n_0+1}^n E_{n_0+1}^n)_{\rho} \, . \label{eq_araki_lieb}
\end{align}
In the proof of~\cref{fact_Giulia_proof} it is shown that $\rho_{A_1^{n_0} E_1^{n_0}} \in \mathrm{span} \cV(\cH_{AE}^{\otimes n_0},\ket{\theta}_{AE}^{\otimes n_0-r})$. This implies 
\begin{align}
H(A_1^{n_0} E_1^{n_0})_{\rho} 
\overset{\textnormal{\Cshref{eq_sizeT}}}{\leq} \log \left( 2^{n_0 h(r/n_0)} d_{AE}^r \right)
= n_0 h\Big(\frac{r}{n_0} \Big) + r \log d_{AE} \, . \label{eq_new_ds1}
\end{align}
Similarly, we obtain
\begin{align}
H(A_{n_0+1}^n E_{n_0+1}^n)_{\rho}
\leq (n-n_0) h\Big(\frac{r}{n-n_0} \Big) + r \log d_{AE} \, . \label{eq_new_ds2}
\end{align}
Hence we find
\begin{align} 
I(A_1^k : A_{k+1}^{k+k'}|A_{k+k'+1}^m)_{\rho}
\overset{\textnormal{\Cshref{eq_araki_lieb}}}&{\leq}H(A_1^{n_0} B_1^{n_0})_{\rho} + H(A_{n_0+1}^n B_{n_0+1}^n)_{\rho}\\
\overset{\textnormal{\Cshref{eq_new_ds1,eq_new_ds2}}}&{\leq} n h\Big(\frac{2r}{n} \Big) + 2r \log d_{AE}   \, .\label{eq_complicated_proof_1}
\end{align}

For any $\ell < s = \lfloor (\log n)^\frac{3}{2} \rfloor$, we thus have 
\begin{align}
    \norm{\rho_{A_1^{\ell+1}} - \sigma_A^{\otimes (\ell+1)}}_1 
    \leq \norm{\rho_{A_1^{s}} - \sigma_A^{\otimes s}}_1 
    \overset{\textnormal{\Cshref{prop_distance}}}{\leq} \varepsilon_n\,,
\end{align}
where the first step follows since the trace distance is contractive under trace-preserving completely positive maps~\cite[Theorem~8.16]{wolf_notes}.
The continuity of conditional entropy~\cite{AF03} then implies that
\begin{align} \label{eq_complicated_proof_2}
| H(A_1| A_2^{\ell+1})_{\rho} - H(A)_{\sigma} | = |H(A_1| A_2^{\ell+1})_{\rho} -  H(A_1| A_2^{\ell+1} )_{\sigma^{\otimes (\ell+1)}} | \leq 4 \eps_n \log d_A + 2h(\eps_n) =: \delta_n \, ,
\end{align}
where $d_A = \dim(A)$.
The chain rule together with permutation invariance allows us to write 
\begin{align}
H(A_1^n)_\rho 
\overset{\textnormal{chain rule}}&{=} H(A_1^{n-\ell}|A_{n-\ell+1}^n)_\rho + H(A_{n-\ell+1}^n)_\rho \\
&\geq H(A_1^{n-\ell}|A_{n-\ell+1}^n)_\rho \\
\overset{\textnormal{chain rule}}&{=} \sum_{i=1}^{n-\ell} H(A_i | A_{n-\ell+1}^n)_\rho - \sum_{i=1}^{n-\ell} I(A^{i-1}_1:A_i|A_{n-\ell+1}^n)_\rho \\
\overset{\textnormal{perm.~inv.}}&{=} (n-\ell) H(A_1 | A_{2}^{\ell+1})_\rho - \sum_{i=1}^{n-\ell} I(A^{i-1}_1:A_i|A_{n-\ell+1}^n)_\rho \\
\overset{\textnormal{\Cshref{eq_complicated_proof_2}}}&{\geq} (n-\ell) H(A)_{\sigma} - \sum_{i=1}^{n-\ell} I(A^{i-1}_1:A_i|A_{n-\ell+1}^n)_\rho - n\delta_n \\
&=(n\!-\!\ell) H(A)_{\sigma} \!-\! \sum_{i=1}^{n-s} I(A^{i-1}_1:A_i|A_{n-\ell+1}^n)_\rho \!-\! \sum_{i=n-s+1}^{n-\ell} I(A^{i-1}_1:A_i|A_{n-\ell+1}^n)_\rho  \!-\! n \delta_n \\
&\geq (n-\ell) H(A)_{\sigma} - \sum_{i=1}^{n-s} I(A^{i-1}_1:A_i|A_{n-\ell+1}^n)_\rho -2(s-\ell) \log d_A  - n \delta_n \\
\overset{\textnormal{perm.~inv.}}&{=} (n-\ell) H(A)_{\sigma} - \sum_{i=1}^{n-s} I(A^{i-1}_1:A_i|A_{n-s+1}^{n-s+\ell})_\rho -2(s-\ell) \log d_A  - n \delta_n \, . \label{eq_complicated_proof_3}
\end{align}
\begin{claim} \label{claim_CMI}
Let $n \in \N$, exists $\ell \leq s= \lfloor \sqrt{n} \rfloor$ such that for $\sum_{i=1}^{n-s} I(A^{i-1}_1:A_i|A_{n-s+1}^{n-s+\ell})_\rho =: \xi_n$ we have $\lim_{n \to \infty} \frac{\xi_n}{n}=0$.
\end{claim}
\begin{proof}
The proof follows the idea from~\cite[Equation (5)]{berta2023}, which proves a variant of~\cref{claim_CMI} for a purely classical scenario. 
By the permutation-invariance we have for any $1 \leq i < k \leq m \leq n$
\begin{align}
    I(A_1^{i-1}:A_i|A_{k+1}^m)_{\rho} = I(A_1^{i-1}:A_m | A_{k}^{m-1})_{\rho} \, . \label{eq_step_conc1}
\end{align}
This allows us to write
\begin{align}
\sum_{m=k}^n I(A_1^{i-1}:A_i | A_{k+1}^m)_{\rho}
\overset{\eqref{eq_step_conc1}}{=}\sum_{m=k}^n I(A_1^{i-1}:A_m | A_{k}^{m-1})_{\rho}
\overset{\textnormal{chain rule}}{=} I(A_1^{i-1}:A_k^n)_{\rho} \, . \label{eq_step_conc2}
\end{align}
Summing~\cref{eq_step_conc2} over all $1\leq i < k$ and dividing by $n-k+1$ gives
\begin{align}
\frac{1}{n-k+1} \sum_{m=k}^n \sum_{i=1}^{k-1} I(A_1^{i-1}:A_i | A_{k+1}^m)_{\rho}
= \frac{1}{n-k+1} \sum_{i=1}^{k-1} I(A_1^{i-1}:A_k^n)_{\rho} \, .
\end{align}
Hence, there exists $m^* \in \{k,k+1,\ldots,n \}$ such that
\begin{align}
\sum_{i=1}^{k-1} I(A_1^{i-1}:A_i | A_{k+1}^{m^*})_{\rho}
\leq \frac{1}{n-k+1} \sum_{i=1}^{k-1} I(A_1^{i-1}:A_k^n)_{\rho} \, . \label{eq_step_conc3}
\end{align}
Choosing $k=n-s$, \cref{eq_step_conc3} can be rewritten (as there exists $0 \leq \ell \leq s$ such that $m^* = k+\ell$) as
\begin{align}
\sum_{i=1}^{n-s-1} I(A^{i-1}_1:A_i|A_{n-s+1}^{n-s+\ell})_{\rho}
\leq \frac{1}{s+1} \sum_{i=1}^{n-s}I(A_1^{i-1}:A_{n-s}^n)_{\rho} \label{eq_step_GM_1}
\end{align}
We next show that~\cref{eq_complicated_proof_1} implies
\begin{align} \label{eq_step_GM_2}
I(A_1^{i-1}:A_{n-s}^n)_{\rho} \leq  2n h\Big(\frac{2r}{n} \Big) + 4r \log d_{AE} \quad \forall i=1,\ldots,n-s \, .
\end{align}
To see this, we can assume w.l.o.g.\footnote{It can be shown that $s+1 \leq n_0 + 2$ for all $n\in \mathbb{N}$. Hence, if $s+1 = n_0 + 2$, due to the chain rule we can write $I(A_1^{i-1}:A_{n-s}^n)_{\rho} = I(A_1^{i-1}:A_{n-s+1}^n)_{\rho} + I(A_1^{i-1}:A_{n-s}^{n-s}|A_{n-s+1}^n)_{\rho} \leq I(A_1^{i-1}:A_{n-s+1}^n)_{\rho} + 2 \log d_A$, and the argument above still works.} that $n$ is large enough such that 
\begin{align}
k' 
:= s+1
= \left \lfloor (\log n)^{\frac{3}{2}} \right \rfloor +1
\leq \left \lfloor \frac{n}{2} \right \rfloor + 1
= n_0 + 1 \, .
\end{align}
Let us now consider two cases:
In case $i-1\leq n_0$, the assertion follows from~\cref{eq_complicated_proof_1} by choosing $m=k+k'$, since
\begin{align}
    I(A_1^{i-1}:A_{n-s}^n)_{\rho} \overset{\textnormal{perm.~inv.}}&{=} I(A_1^{i-1}:A_{i}^{i+s})_{\rho} \leq n h\Big(\frac{2r}{n} \Big) + 2r \log d_{AE} \,.
\end{align}
In case $i-1>n_0$, we can use the chain rule to write
\begin{align}
I(A_1^{i-1}:A_{n-s}^n)_{\rho} &= I(A_1^{n_0}:A_{n-s}^n)_{\rho} + I(A_{n_0+1}^{i-1}:A_{n-s}^n|A_1^{n_0})_{\rho} \\
\overset{\textnormal{perm.~inv.}}&{=} I(A_1^{n_0}:A_{{n_0}+1}^{{n_0}+1+s})_{\rho} + I(A_{1}^{i-1-{n_0}}:A_{i-{n_0}}^{i-{n_0}+s}|A_{i-n_0+s+1}^{i+s})_{\rho} \, ,
\end{align}
where both terms on the right-hand side are bounded from above by $n h(\frac{2r}{n}) + 2r \log d_{AE} $ via~\cref{eq_complicated_proof_1}.

Putting things together yields
\begin{align}
\sum_{i=1}^{n-s} I(A^{i-1}_1:A_i|A_{n-s+1}^{n-s+\ell})_{\rho}
&= \sum_{i=1}^{n-s-1} I(A^{i-1}_1:A_i|A_{n-s+1}^{n-s+\ell})_{\rho} + I(A_1^{n-s-1}:A_{n-s}|A_{n-s+1}^{n-s+\ell})_{\rho}\\
\overset{\textnormal{\Cshref{eq_step_GM_1}}}&{\leq} \frac{1}{s+1} \sum_{i=1}^{n-s}I(A_1^{i-1}:A_{n-s}^n)_{\rho} + 2 \log d_A\\
\overset{\textnormal{\Cshref{eq_step_GM_2}}}&{\leq} \frac{n-s}{s+1} \left(2n h\Big(\frac{2r}{n} \Big) + 4r \log d_{AE} \right) + 2 \log d_A \\
&=o(n) \, ,
\end{align}
where the final step uses that for $s=\lfloor \sqrt{n} \rfloor$ and $r=o(n)$ we have
\begin{align}
\lim_{n \to \infty} \frac{n}{s} h\Big(\frac{2r}{n} \Big) = 0 \, .
\end{align}
\end{proof}
Since $\lim_{n \to \infty} \delta_n =0$, combining~\cref{claim_CMI} with~\cref{eq_complicated_proof_3} yields for some $\ell \leq s = \lfloor \sqrt{n} \rfloor$
\begin{align}
H(A_1^n)_\rho  
\geq (n-\ell) H(A)_{\sigma} - o(n)
=n H(A)_{\sigma}  - o(n)\, .
\end{align}
\qed
%%%%%%%%%%%%%%%%%%%%%%%%%%%%%%%%%%%%%%%%%%%%%%%%%%%%%%%%%%%%%%%%%%%%
%%%%%%%%%%%%%%%%%%%%%%%%%%%%%%%%%%%%%%%%%%%%%%%%%%%%%%%%%%%%%%%%%%%%
%%%%%%%%%%%%%%%%%%%%%%%%%%%%%%%%%%%%%%%%%%%%%%%%%%%%%%%%%%%%%%%%%%%%
\bibliographystyle{arxiv_no_month}
\bibliography{bibliofile}

\end{document}